\documentclass[letterpaper,journal]{IEEEtran}
\usepackage{amsmath,amssymb,amsfonts,amsthm}
\usepackage{mathtools}
\usepackage{array}
\usepackage{textcomp}
\usepackage{stfloats}
\usepackage{url}
\usepackage{verbatim}
\usepackage{graphicx}
\graphicspath{{./}{../../}}
\usepackage{cite}
\usepackage[caption=false,font=footnotesize]{subfig}
\usepackage{hyperref}
\def\BibTeX{{  B\kern-.05em{\sc i\kern-.025em b}\kern-.08em
    T\kern-.1667em\lower.7ex\hbox{E}\kern-.125emX}}
\hypersetup{colorlinks=false,
            linkcolor=red,
            anchorcolor=green,
            citecolor=black}

\usepackage[usenames]{color}

\newtheorem{Lemma}{Lemma}
\newtheorem{Theorem}{Theorem}

\newtheorem{Proposition}{Proposition}
\newtheorem{Remark}{Remark}

\begin{document}

% \title{Realizing Wireless Distributed MoE: An Over-the-Air Computation Approach
% }

\title{
AirMoE: 
Realizing Over-the-Air Distributed Mixture-of-Experts Inference at the Wireless Edge
}

\author{{  {Huiling~Yang}, {Zhanwei~Wang}, and {Kaibin~Huang}                   }

\thanks{H.~Yang, Z.~Wang, and K.~Huang are with the Department of Electrical and Computer Engineering, The University of Hong Kong (HKU), Hong Kong SAR, China (Email: \{hlyang, zhanweiw,  huangkb\}@eee.hku.hk).
 Corresponding authors: K.~Huang and Z.~Wang.
}
}
\maketitle

\begin{abstract}

\emph{Mixture-of-experts} (MoE) architectures enable efficient \emph{large language model} (LLM) inference at the wireless edge through sparse activation. 
The \emph{wireless distributed MoE} (WIDE) architecture addresses edge-resource constraints by distributing experts across devices coordinated by an edge server. 
However, WIDE suffers from repeated uplink transmissions of high-dimensional expert outputs via orthogonal multiple access. 
To overcome this bottleneck, we propose \emph{AirMoE}, an \emph{over-the-air computing} (AirComp)-enabled framework for simultaneous expert-output aggregation via wireless waveform superposition.
Integrating AirComp into MoE inference introduces three challenges: fast-varying aggregation weights, layer-dependent error sensitivity, and channel-aware expert placement. 
To address these challenges, we construct an inference-aware AirMoE error metric to quantify aggregation distortion effects on \emph{end-to-end} (E2E) inference accuracy via perturbation-based layer-sensitivity calibration.
We then formulate a joint optimization problem to minimize this error and decompose it, without loss of optimality, into a two-timescale framework. 
At the fast timescale, we derive a globally optimal threshold-based power-control policy that separates devices into coefficient-aligned and full-power groups.  
At the slow timescale, we develop an activation- and channel-aware expert placement strategy that assigns more important experts to devices with lower channel-power cost.
Extensive experiments demonstrate that AirMoE outperforms representative baselines in E2E inference accuracy, especially under strong device heterogeneity.

\end{abstract}


\section{Introduction}
\label{sec:introduction}

The \emph{sixth-generation} (6G) mobile networks will support the deployment of \emph{large language models} (LLMs) at the network edge to bring intelligent services closer to users. 
They will provide reasoning capabilities to low-latency intelligent applications, such as virtual reality, autonomous systems, and personalized assistants~\cite{letaief2019roadmap,friha2024llm}.
However, edge deployment of LLMs remains challenging, as their substantial computation and memory demands often exceed the capacity of individual edge nodes~\cite{zhou2019edge,zheng2026multi}.
A popular LLM architecture, known as mixture-of-experts (MoE), expands model capacity through multiple expert modules with per-token sparse expert activation~\cite{lepikhin2021gshard,fedus2022switch}.
To address the resource constraints of edge deployment, \emph{wireless distributed MoE} (WIDE) frameworks have emerged to enable distributed implementation by exploiting the modular and parallel structure of MoE~\cite{chen2026siftmoe,xue2025wdmoe}.
Specifically, computation-intensive expert modules (i.e., sub-networks) are distributed across edge devices to leverage device-side memory and computation resources, while the remaining model components are placed at the edge server. 
Since MoE inference proceeds through a cascade of layers, such distributed deployment requires repeated server--device interactions. 
During the execution of each layer, the server dispatches hidden states and routing information to activated expert devices, and the devices wirelessly upload their expert outputs to the server for aggregation. 
Existing WIDE designs based on orthogonal multiple access suffer from an uplink communication bottleneck, which arises from the upload of high-dimensional outputs from multiple expert devices. 
To address this issue, we propose \emph{Over-the-Air} MoE (AirMoE), an \emph{over-the-air computing} (AirComp) enabled distributed MoE inference framework that accelerates inference by realizing simultaneous multi-expert access while minimizing channel-induced aggregation errors.

Existing WIDE studies have explored mitigation strategies for the uplink bottleneck under separate digital transmission. 
Representative approaches include reducing the number of activated experts, substituting low-utility experts with more communication-efficient alternatives, and optimizing radio-resource allocation among participating devices~\cite{chen2026siftmoe,xue2025wdmoe,qin2025optimal}. 
Although these designs reduce the upload burden and improve resource utilization, they still require activated experts to compete for finite radio resources due to the assumption of orthogonal multi access. 
The uplink aggregation latency increases approximately with the number of activated experts when their output dimensions and transmission rates are comparable.
For example, state-of-the-art MoE models such as Qwen3-MoE and DeepSeek-V3 route eight experts per token~\cite{yang2025qwen3,deepseekv3}. 
In view of prior work, existing resource management strategies can alleviate but cannot eliminate the uplink bottleneck inherent to orthogonal access.

AirComp provides a promising mechanism to address this limitation~\cite{goldenbaum2013robust}. 
By exploiting the waveform superposition property of a multiple-access channel, AirComp enables concurrent device transmissions and allows the receiver to directly obtain a desired aggregate function, e.g., averaging, weighted summation, and maximization~\cite{zhu2018mimo}. 
The communication efficiency of AirComp has motivated extensive research on wireless edge applications, including federated learning~\cite{wang2024spectrum,sery2021over}, distributed sensing~\cite{wang2026airbreath,liu2023over}, and more recently distributed LLM inference~\cite{zhang2025llm}. 
Nevertheless, AirComp's reliance on uncoded analog transmission makes over-the-air aggregation susceptible to channel fading and interference. 
A broad range of techniques has been developed to reduce the aggregation error through transceiver design and radio resource management, including power control~\cite{cao2020optimized}, beamforming~\cite{chen2018uniform}, interference management ~\cite{cao2021cooperative}, and broadband subcarrier allocation~\cite{qin2021broadband}.
Despite this progress, applying AirComp to distributed MoE inference remains unexplored. 
The motivation for investigating this direction lies in the need for WIDE systems to combine high-dimensional outputs of activated experts when executing each MoE layer.

While generic AirComp provides a promising mechanism, realizing highly accurate and efficient AirComp-enabled MoE inference requires a task-oriented design, namely AirMoE, tailored to address three unique challenges.
The first is \emph{fast-varying gating weights}. 
Unlike conventional AirComp that assumes fixed or slowly varying aggregation coefficients~\cite{goldenbaum2013robust,zhu2018mimo}, MoE inference produces token- and layer-dependent gating scores. 
These fast-varying weights, which change at the token timescale and are tightly coupled with time-varying channels, determine the activated expert subset and their relative contributions. 
This requires the AirComp transceiver to adapt instantaneously upon the input of each token.
The second is \emph{layer-dependent error sensitivity}. Conventional AirComp minimizes the physical-layer \emph{mean squared error} (MSE)~\cite{cao2020optimized,chen2018uniform,liu2020over,zhu2020broadband}, a metric that falls short of capturing \emph{end-to-end} (E2E) MoE inference accuracy~\cite{wen2024task,yang2026mimo,wang2026revisiting}. 
The recovered aggregate at each layer is an intermediate representation whose distortion propagates forward and perturbs subsequent routing decisions~\cite{qiu2025layerwise}. 
Consequently, aggregation errors of equal magnitude can have drastically different effects on the final prediction depending on the layer in which they occur. 
This necessitates an AirComp design that explicitly accounts for such layer-dependent error sensitivity.
The third is \emph{distortion-aware expert placement}, a challenge unique to distributed MoE and absent in standard AirComp designs. 
Wireless edge devices typically exhibit heterogeneous communication and computation capabilities~\cite{yang2026optimal}. 
E2E accuracy fundamentally depends on the assignment of experts to these heterogeneous devices. 
Once activated, an expert's assigned device dictates the channel and power constraints for transmitting its output over the air. 
Therefore, assigning frequently activated experts to devices with weak channels causes persistent aggregation errors. 
While existing MoE placement schemes primarily optimize latency, workload balance, or communication overhead~\cite{go2025moetuner,li2026optimizing,sivtsov2025cluster,chen2026slimcaching,wang2026spacemoe}, AirMoE introduces a novel objective that strategically couples the inference importance of individual experts with device channel heterogeneity.

Addressing these challenges motivates the design of the AirMoE framework. 
AirMoE adopts a one-expert-per-device deployment as the primary system setting to accommodate strict memory constraints at devices and enable parallel expert computation, while we discuss a multi-expert-per-device deployment as an extension. 
Under this setting, we introduce a novel inference-aware AirMoE error metric and develop a two-timescale optimization framework. Specifically, this design optimizes instantaneous power control to adapt to fast-varying gating weights (fast timescale), alongside strategic expert placement (slow timescale). 
The unified objective is to realize simultaneous over-the-air expert-output aggregation while minimizing the resulting degradation in E2E inference accuracy.
Our main contributions and findings are summarized as follows:

\begin{itemize}

    \item \textbf{AirMoE Framework:} We present an AirMoE inference protocol that maps the gating-weighted sum in MoE expert aggregation onto wireless waveform superposition, thereby enabling activated expert outputs to be aggregated directly over the air. To characterize the reliability of AirMoE inference, we define a layer-wise AirMoE error to quantify each aggregation distortion and construct an overall AirMoE error metric through perturbation-based layer-sensitivity calibration. We then formulate a joint power-control and expert-placement problem to minimize the overall AirMoE error. By exploiting the inherent two-timescale decision structure, we decompose this joint problem without loss of optimality into instantaneous power-control and long-term placement subproblems to simplify solution computation.  

    \item \textbf{Optimal Power Control for AirMoE:}
    For each instantaneous AirMoE aggregation, we solve the first subproblem to jointly optimize the device transmit powers and the server denoising factor to minimize the layer-wise AirMoE error. Thereby, we derive a globally optimal threshold-based policy governed by the ratio between each device's channel–power capability and its target gating weight. Specifically, devices satisfying this threshold achieve exact aggregation-weight alignment, whereas the remaining devices transmit at maximum power to minimize the alignment gap. Furthermore, asymptotic analysis demonstrates that for all activated devices, the optimal policy reduces to gating-weighted channel inversion in the high signal-to-noise ratio (SNR) regime and full-power transmission in the low-SNR regime.

    \item \textbf{Activation- and Channel-Aware Expert Placement:}
    Under the preceding optimal power-control policy, we solve the associated subproblem to optimize the long-term expert-to-device mapping and minimize the AirMoE error. For mathematical tractability, we derive an upper bound on the placement objective to serve as a surrogate, which factorizes neatly into two decoupled components: expert importance and device cost. The expert importance encapsulates layer sensitivity, activation frequency, gating weight, and expert-output scale, whereas the device cost characterizes the long-term inverse channel--power capability of each device. Optimizing this surrogate yields a simple placement rule that matches experts in descending order of importance with devices in ascending order of their cost.

    \item \textbf{Experimental Results:}
    Experiments with OLMoE-1B-7B-0924, a popular MoE model, on the language-reasoning dataset under heterogeneous wireless channels validate the proposed AirMoE error as a task-relevant distortion metric. The results show that AirMoE outperforms representative power-control and placement benchmarks in E2E inference accuracy, with larger gains under stronger device heterogeneity.

\end{itemize}

The remainder of this paper is organized as follows. Section~\ref{sec:system_models} introduces the system and communication models, followed by the AirMoE inference protocol and uplink aggregation design in Section~\ref{sec:overview_airmoe_inference}. Section~\ref{sec:airmoe_problem} calibrates layer sensitivity and formulates the two-timescale joint power-control and expert-placement problem.
The optimal power-control policy is derived in Section~\ref{sec:power_control_minimizing_airmixer_error}, while Section~\ref{sec:joint_expert_placement_power_control}  develops the activation- and channel-aware expert-placement strategy. Section~\ref{sec:experimental_results} provides the experimental results, and Section~\ref{sec:conclusion} concludes the paper.

\section{System Models}
\label{sec:system_models}

% \begin{figure}[t]
%     \centering
%     \subfloat[MoE architecture\label{fig:moe_architecture}]{
%         \includegraphics[width=0.47\linewidth]{Images/moe_architecture.pdf}
%     }
%     \hfill
%     \subfloat[Wireless MoE deployment\label{fig:moe_deployment}]{
%         \includegraphics[width=0.47\linewidth]{Images/wireless_distributed_experts.pdf}
%     }
%     \caption{Wireless distributed expert system. }
%     \label{fig:wireless_distributed_expert_model}
%     \vspace{-3mm}
% \end{figure}

We consider a wireless distributed MoE system in which expert networks are deployed across edge devices and coordinated by an edge server to perform token generation services. The distributed architecture for MoE inference and the associated communication patterns are described as follows.

\subsection{Wireless Distributed Expert Model}
\label{subsec:wireless_distributed_expert_model}

\subsubsection{Distributed MoE Architecture}

We first introduce the MoE modules and then describe how they are mapped onto edge networks. As shown in Fig.~\ref{fig:moe_architecture}, the MoE model consists of an embedding layer, $L$ cascaded MoE layers indexed by $\mathcal{L}=\{1,\dots,L\}$, and a language modeling head. Each MoE layer is decomposed into a \emph{gateway module} and an \emph{expert module}.
The gateway module contains the non-expert operations, including self-attention, residual connections, layer-normalization operations, and the gating network. 
It produces the contextual hidden representation and computes the gating scores for expert selection. 
The expert module consists of $I$ parallel experts, i.e., \emph{feed-forward networks} (FFNs), indexed by $\mathcal{I}=\{1,\dots,I\}$.

As shown in Fig.~\ref{fig:airmoe_inference_protocol}, AirMoE maps the MoE modules onto a wireless edge system consisting of an edge server and a set of available wireless devices $\mathcal M=\{1,\dots,M\}$. Specifically, the edge server hosts the $L$ gateway modules and the language modeling head, whereas the embedding layer and experts are deployed across distributed devices. Due to device memory limitations and the need for parallel processing, we consider a one-to-one mapping with $M\ge LI$, where each expert in each MoE layer is hosted by one device~\cite{xue2025wdmoe,chen2026siftmoe,wang2026spacemoe}. To avoid the excessive communication overhead caused by expert migration, the expert-device assignment remains fixed during online inference, with its optimization studied in Section~\ref{sec:joint_expert_placement_power_control}.
% The communication process between the edge server and the devices is specified in Section~\ref{subsec:communication_model}.

\begin{figure}[t]
    \centering
    \includegraphics[width=0.7\linewidth]{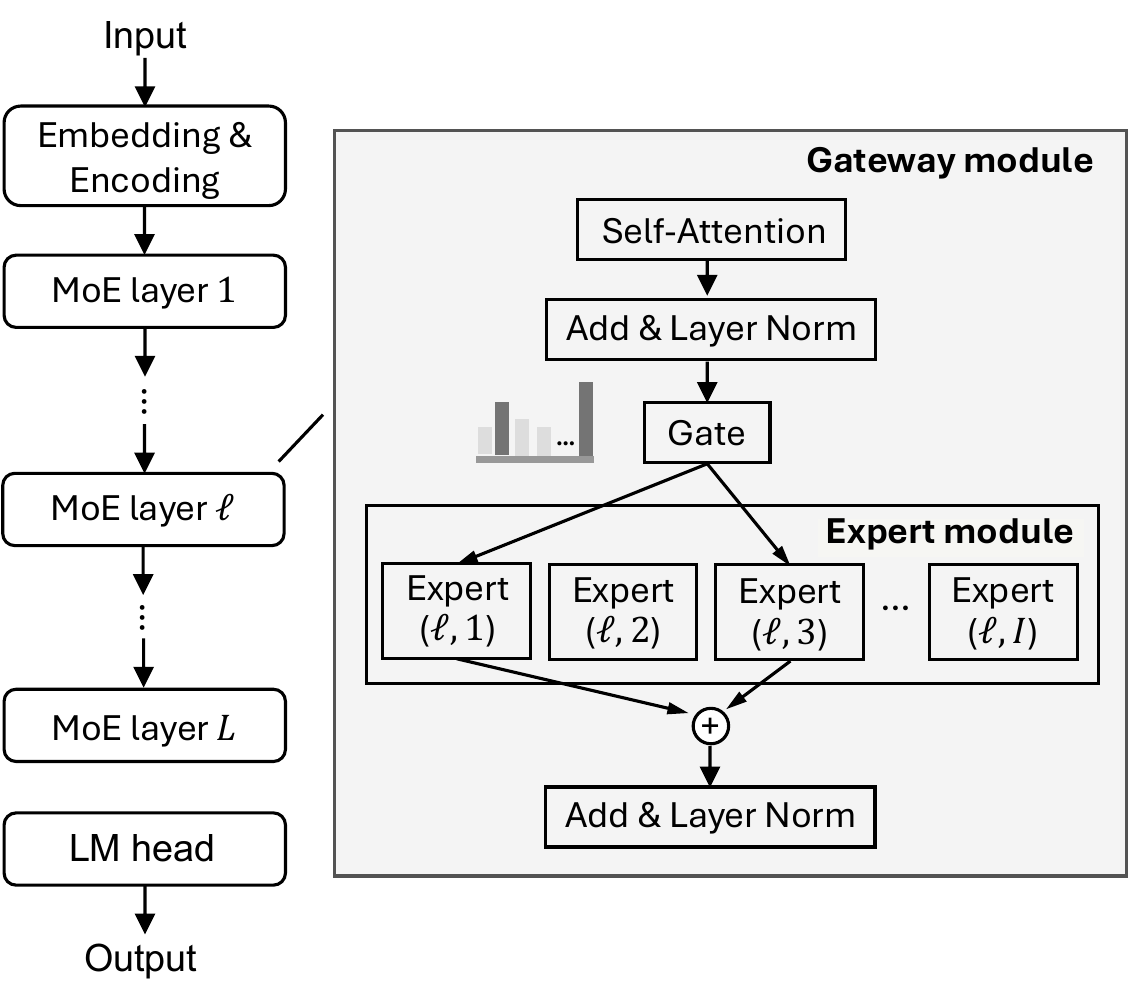}
    \caption{
    MoE model architecture. 
    }
    \label{fig:moe_architecture}
    \vspace{-3mm}
\end{figure}

\subsubsection{Expert Activation and Output Aggregation Model}

Consider that the generation of each token requires propagation through all MoE layers.
We next characterize the expert activation and output aggregation process for an arbitrary token and MoE layer.
The layer and token indices $(\ell,t)$ are omitted for notational simplicity.
Let $\mathbf{u}\in\mathbb{R}^{D}$ denote the hidden state obtained after the self-attention and layer-normalization operations, where $D$ is the hidden dimension. The gating network maps $\mathbf{u}$ to the routing-weight vector $\mathbf{g}=[g_1,\dots,g_I]^{\mathrm{T}}\in\mathbb{R}^{I}$ as \begin{equation} \label{eq:gating_score} \mathbf{g} \triangleq \mathrm{softmax}\!\left(\mathbf{W}_{g}\mathbf{u}\right), \end{equation} where $\mathbf{W}_{g}\in\mathbb{R}^{I\times D}$ is the router parameter matrix, and 
$g_i>0$ denotes the gating score of expert $i$, with $\sum_{i=1}^{I}g_i=1$.
Under a Top-$K$ routing strategy, the $K$ experts with the largest gating scores are activated.
The corresponding activated expert set is denoted by
$S\subseteq\mathcal{I}$, with $|S|=K$.

Each activated expert $i\in S$ computes its local expert output as
$\mathbf{v}_i=\mathrm{FFN}_i(\mathbf{u})\in\mathbb{R}^{D}$.
The expert aggregation output is then obtained by the gating-weighted aggregation of the activated expert outputs\footnote{We adopt the original softmax gating scores in \eqref{eq:moe_ideal_agg}, rather than re-normalizing them over the Top-$K$ selected experts, following recent fine-grained MoE models such as OLMoE~\cite{muennighoff2025olmoe} and Qwen2-MoE~\cite{team2024qwen2}.}, given as
\begin{equation}
\label{eq:moe_ideal_agg}
\mathbf{o}
=
\sum_{i\in S}
g_i
\mathbf{v}_i\in\mathbb{R}^{D}.
\end{equation}
The aggregated output is passed through the subsequent residual and normalization operations before entering the next layer. The ideal aggregation in \eqref{eq:moe_ideal_agg} serves as the reference target for the wireless aggregation mechanism developed in the following sections.

\subsection{Communication Model}
\label{subsec:communication_model}

% \begin{figure*}[t]
%     \centering
%     \includegraphics[width=0.8\linewidth]{Images/airmoe_inference.pdf}
%     \caption{AirMoE inference protocol.}
%     \label{fig:airmoe_inference_protocol}
%     \vspace{-3mm}
% \end{figure*}

In the wireless distributed MoE system, the AirComp technique is considered for approximating the target aggregation in \eqref{eq:moe_ideal_agg}.
Next, we define the relevant symbols and waveform superposition functions.
For an arbitrary expert aggregation with activated expert set $S$, the expert-device mapping induces the corresponding activated device set $\mathcal A\subseteq\mathcal M$.
Each device $m\in\mathcal A$ preprocesses its local expert output into a symbol vector $\mathbf{s}_m\in\mathbb{R}^{D}$,  which spans over $D$ consecutive symbol durations.
We adopt a block-fading uplink channel model, where the channel coefficient of each activated device remains constant during a layer-wise expert aggregation.
The uplink channel coefficient from device $m$ to the edge server is denoted by $h_m$.
The \emph{channel state information} (CSI) is assumed to be available at both the edge server and the transmitting devices~\cite{liu2023over}. 
Under symbol-level synchronization guaranteed by timing advance~\cite{wang2024spectrum,3gpp36213}, all devices in $\mathcal A$  simultaneously upload their expert outputs, and the received signal vector at the edge server is given by
\begin{equation}
\label{eq:uplink_received_signal}
\mathbf{y}
=
\sum_{m\in\mathcal A}
h_m b_m \mathbf{s}_m
+
\mathbf{z},
\end{equation}
where $b_m$ denotes the scalar precoding coefficient of device $m$ and $\mathbf{z}$ denotes the complex additive white Gaussian noise (AWGN).

On the downlink, the server can broadcast the hidden state and routing information to the activated devices through digital or analog transmission~\cite{wang2024spectrum}. 
Given sufficient transmit power at the edge server, the downlink transmission is more reliable than uplink expert-output uploading, and its distortion is therefore neglected in the subsequent analysis.

\begin{figure}[t]
    \centering
    \includegraphics[width=0.8\linewidth]{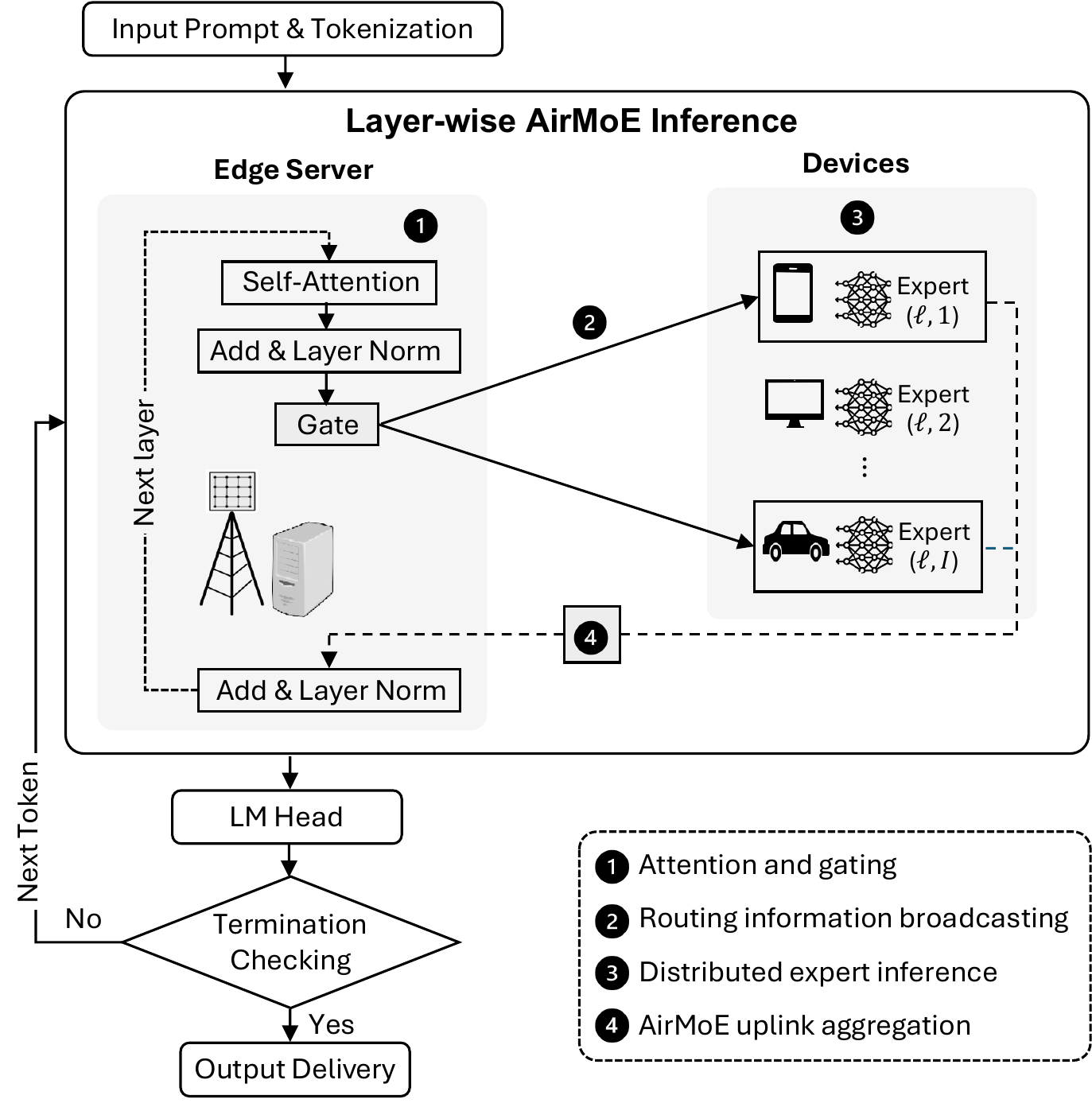}
    \caption{
    AirMoE inference protocol. 
    }
    \label{fig:airmoe_inference_protocol}
    \vspace{-3mm}
\end{figure}

\section{Overview of AirMoE Inference}
\label{sec:overview_airmoe_inference}
In this section, we present an overview of the AirMoE framework, including the inference protocol, over-the-air aggregation, and performance metric.

\subsection{AirMoE Inference Protocol}
\label{subsec:airmoe_inference_protocol}

Based on the wireless distributed MoE architecture, the AirMoE inference protocol is illustrated in Fig.~\ref{fig:airmoe_inference_protocol}. The requesting device first tokenizes the input prompt and obtains the initial token representations through the embedding layer. These representations are transmitted to the edge server, which coordinates the subsequent distributed MoE inference.

For each generated token, AirMoE processes the current representation sequentially through the $L$ MoE layers.
At each MoE layer, the processing proceeds as follows. 
\begin{enumerate}
    \item \textit{Server-side attention and gating:} 
    The edge server performs the gateway computation, in which the self-attention and layer-normalization operations generate the hidden state $\mathbf{u}$, and the gating network computes the gating scores $\{g_i\}_{i\in\mathcal I}$ for all experts. 
    Based on these scores, the edge server selects the Top-$K$ experts and obtains the activated expert set $S$.
    \item \textit{Routing information broadcasting:}
    According to the expert-device mapping, the edge server identifies the devices hosting the activated experts in $S$ and broadcasts the hidden state $\mathbf{u}$ together with the corresponding gating scores $\{g_i\}_{i\in S}$ to these devices.

    \item \textit{Distributed expert inference:}
    Upon receiving the hidden state, each activated device executes its local expert in parallel and obtains the expert output
    $\mathbf{v}_i=\mathrm{FFN}_i(\mathbf{u})$.

    \item \textit{AirMoE uplink aggregation:}
    The activated devices simultaneously transmit their expert outputs to the edge server to realize the gating-weighted aggregation in \eqref{eq:moe_ideal_agg} over the air. The corresponding transceiver design is detailed in the next subsection.
\end{enumerate}

After the token representation is processed by the $L$-th MoE layer, the edge server feeds it into the language-modeling head to produce the next token. The server then checks the predefined stopping rule, such as whether an end-of-sequence token has been generated or the maximum output length has been reached. 
If the stopping rule is not satisfied, the newly generated token is appended to the input sequence, and AirMoE repeats the above inference procedure for the subsequent token.
Otherwise, the edge server performs the required post-processing, e.g., detokenization or formatting, and delivers the generated response to the requesting device.

\begin{figure}[t]
    \centering
    \includegraphics[width=0.9\linewidth]{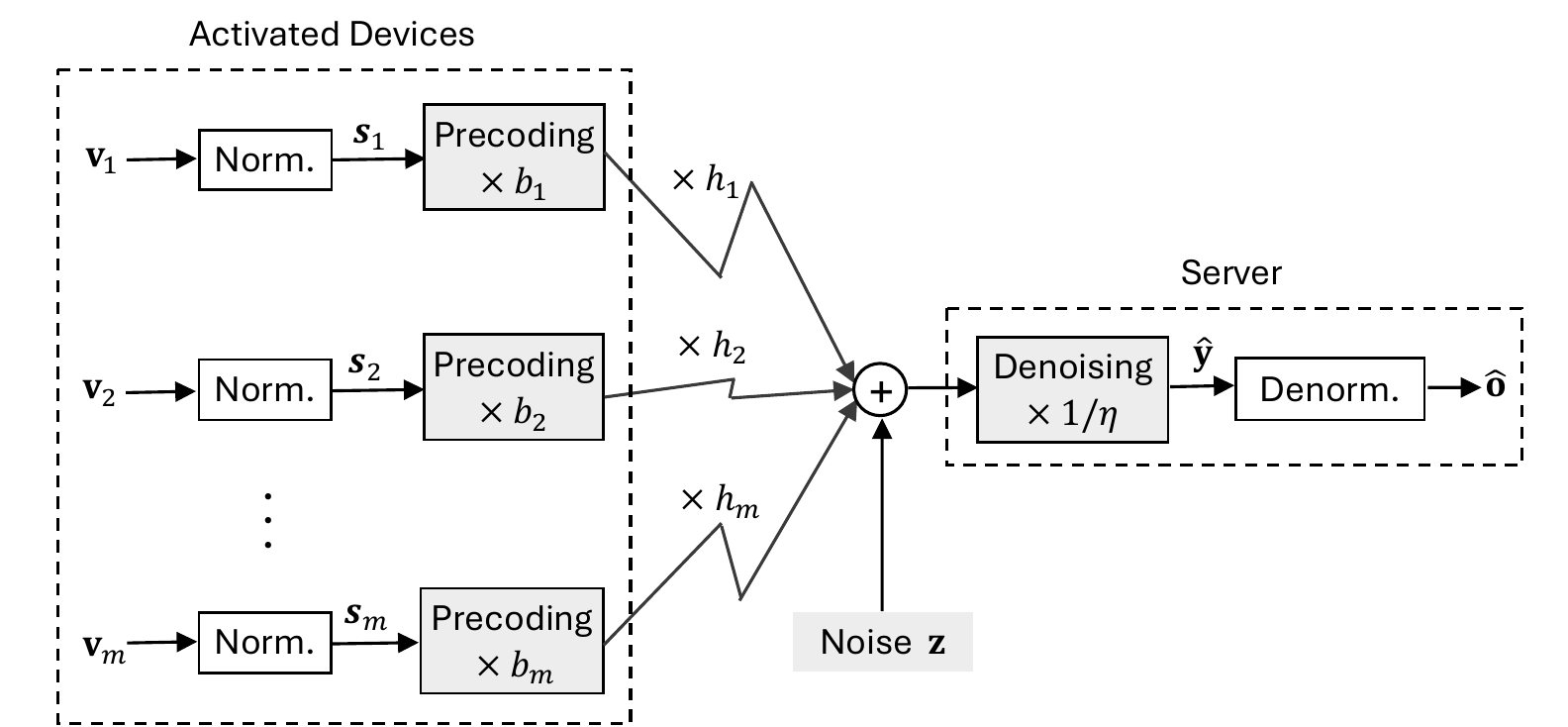}
    \caption{
    AirMoE uplink aggregation. 
    }
    \label{fig:airmoe_uplink_aggregation}
    \vspace{-3mm}
\end{figure}

\subsection{AirMoE Uplink Aggregation}
\label{subsec:airmixer_uplink_ota_expert_aggregation}

The gating-weighted aggregation of expert outputs is performed through uplink AirMoE aggregation as illustrated in Fig.~\ref{fig:airmoe_uplink_aggregation}. The device-side and server-side operations are detailed as follows.

% \begin{figure}[t]
%     \centering
%     \includegraphics[width=0.9\linewidth]{Images/airmoe_uplink_aggragation.pdf}
%     \caption{Uplink AirMoE aggregation process.}
%     \label{fig:airmixer}
%     \vspace{-3mm}
% \end{figure}

\subsubsection{Device-Side Normalization and Precoding}
\label{subsubsec:device_side_precoding}

% The proposed AirMoE aggregation extends the AirComp model in \eqref{eq:uplink_received_signal} with expert-output preprocessing, server-side denoising, and output reconstruction. 
% Through uplink waveform superposition, AirMoE realizes the gating-weighted expert aggregation in \eqref{eq:moe_ideal_agg} as shown in Fig.~\ref{fig:airmixer}.
% The detailed design is presented below.

Consider one expert aggregation at MoE layer $\ell$, with an activated expert set $S$ and the corresponding activated device set $\mathcal A$. 
For each activated device $m\in\mathcal A$, let $\mathbf v_m$ and $g_m$ denote the expert output and the associated gating score, respectively. Then, the ideal aggregation in \eqref{eq:moe_ideal_agg} can be equivalently written as
\begin{equation}
\label{eq:moe_ideal_device_form}
\mathbf{o}
=
\sum_{m\in\mathcal A}
g_m
\mathbf v_m
\in\mathbb{R}^{D}.
\end{equation}

Following the literature~\cite{wang2024spectrum,wang2026airbreath,liu2023over}, to facilitate transmit power control, AirMoE applies a layer-wise normalization to the expert outputs.
Specifically, device $m\in\mathcal A$ constructs the transmit symbol vector as
\begin{equation}
\label{eq:layer_norm_affine}
\mathbf{s}_m \triangleq \frac{ \mathbf v_m - \boldsymbol{\mu}_{\ell} } {c_{\ell}},
\end{equation}
where $\boldsymbol{\mu}_{\ell}=\mathbb{E}[\mathbf v_m]$ and 
$c_{\ell}^{2}
=
\frac{1}{D}
\mathbb{E}\!\left[
\left\|
\mathbf v_m
-
\boldsymbol{\mu}_{\ell}
\right\|_{2}^{2}
\right]$
denote the layer-wise statistical mean and variance of $\mathbf v_m$, respectively.
The expectation is taken over tokens and activated experts at layer $\ell$.
These statistics are estimated offline from representative training samples and shared with the corresponding devices before inference.
The resulting transmit symbol vectors are normalized as
$\mathbb{E}\!\left[\mathbf{s}_m\right]=\mathbf{0}$, $\frac{1}{D}
\mathbb{E}\left[
\|\mathbf{s}_m\|_2^2
\right]
=1.$

To cope with fading channels, device $m\in\mathcal A$ applies a scalar precoder $b_m$ to its normalized symbol vector. With CSI available at the transmitting devices, AirMoE jointly compensates for the channel phase and controls the transmit power by designing the scalar precoder as
\begin{equation}
\label{eq:phase_compensated_precoder}
b_m
=
\sqrt{p_m}\frac{h_m^{\dagger}}{|h_m|},
\qquad
m\in\mathcal A,
\end{equation}
where $p_m\ge 0$ denotes the transmit power of device $m$ and $(\cdot)^\dagger$ denotes the complex conjugate.
This precoder yields the real non-negative effective coefficient $h_m b_m=|h_m|\sqrt{p_m}$ at the edge server.
The transmission of each device is subject to a per-symbol transmit-power constraint $P_m$. Under the unit-power normalization of $\mathbf{s}_m$, this constraint becomes 
\begin{equation} \label{eq:power_constraint} \frac{1}{D} \mathbb{E} \left[ \left\| b_m\mathbf{s}_m \right\|_2^2 \right] = |b_m|^2 = p_m \le P_m, \qquad m\in\mathcal A . \end{equation}

\subsubsection{Server-side Denoising and Reconstruction}
\label{subsubsec:server_side_denoising}

After normalization and precoding,
all activated devices transmit their symbol vectors simultaneously over the uplink channel.
Following the AirComp model in \eqref{eq:uplink_received_signal}, the edge server then applies the denoising factor $\eta$ and extracts the real aggregation component as
\begin{equation}
\label{eq:denoise}
\hat{\mathbf y}
=
\operatorname{Re}
\left\{
\frac{\mathbf y}{\eta}
\right\}
=
\sum_{m\in\mathcal A} \frac{|h_m|\sqrt{p_m}}{\eta} \mathbf s_m + \frac{\mathbf z_R}{\eta},
\end{equation}
where $\mathbf z_R\triangleq\operatorname{Re}\{\mathbf z\}$ denotes the effective real-domain receiver noise with per-dimension variance $\sigma^2$.
The denoising factor controls the scale of the received superposition and the effective noise level after aggregation.

The denoised aggregation $\hat{\mathbf{y}}$ is then mapped back to the original expert-output domain by reversing the affine normalization:
\begin{equation}
\label{eq:reconstruct_x}
\hat{\mathbf{o}} = \left( \sum_{m\in\mathcal A} g_m \right) \boldsymbol{\mu}_{\ell} + c_{\ell} \hat{\mathbf{y}}.
\end{equation}
The reconstructed output $\hat{\mathbf{o}}$ serves as the AirMoE approximation of the ideal MoE expert-block output and is passed to the subsequent residual and normalization operations before entering the next layer.

\subsection{Layer-wise AirMoE Error}
\label{subsec:airmixer_error}

To quantify the distortion introduced by AirMoE aggregation at layer $\ell$, we define the layer-wise AirMoE error as the per-dimension \emph{mean squared error} (MSE) between the reconstructed aggregation output $\hat{\mathbf{o}}$ in \eqref{eq:reconstruct_x} and the ideal aggregation output $\mathbf{o}$ in \eqref{eq:moe_ideal_agg}, given by
\begin{equation}
\begin{aligned}
\label{eq:mse}
\mathrm{MSE}_{\ell}
&\triangleq
\frac{1}{D}
\mathbb{E}
\left[
\left\|
\hat{\mathbf{o}}
-
\mathbf{o}
\right\|_2^2
\right]= \frac{c_{\ell}^{2}}{D} \mathbb{E} \left[ \left\| \hat{\mathbf{y}} - \sum_{m\in\mathcal A} g_m\mathbf{s}_m \right\|_2^2 \right] \\
&=
\frac{c_{\ell}^{2}}{D} \mathbb{E} \Bigg[ \Bigg\| \sum_{m\in\mathcal A} \left( \frac{|h_m|\sqrt{p_m}}{\eta} - g_m \right) \mathbf{s}_m + \frac{\mathbf{z}_R}{\eta} \Bigg\|_2^2 \Bigg],
\end{aligned}
\end{equation}
where the expectation is taken over the normalized expert-output symbols and receiver noise, conditioned on the activated device set, gating scores, channel realization, and AirMoE control variables. As revealed by the second equality, the AirMoE aggregation $\hat{\mathbf y}$ approximates the ideal expert aggregation in the normalized domain, denoted by
$\mathbf  y_{tar}
\triangleq
\sum_{m\in\mathcal A} g_m\mathbf s_m$.
Using the unit-power normalization of $\mathbf{s}_m$, 
the independence between expert-output symbols and receiver noise, and assuming uncorrelated normalized symbols across activated experts, \eqref{eq:mse} reduces to
\begin{equation}
\label{eq:closed_form_mse}
\mathrm{MSE}_{\ell}
=
c_\ell^2
\Bigg(
\underbrace{
\sum_{m\in\mathcal A}
\left(
\frac{|h_m|\sqrt{p_m}}{\eta}
-
g_m
\right)^2
}_{\text{coefficient mismatch}}
+
\underbrace{
\frac{\sigma^2}{\eta^2}
}_{\text{effective noise}}
\Bigg).
\end{equation}
The first term in \eqref{eq:closed_form_mse} measures the mismatch between the effective wireless aggregation coefficients and the desired gating scores, while the second term represents the receiver noise after denoising.

In practical MoE inference, the normalized expert-output symbols of different activated experts may be statistically correlated, leading to additional cross-correlation terms in the exact MSE expansion.
Following the AirComp analysis in~\cite{cao2020optimized}, these terms can be bounded via Cauchy's inequality, resulting in an upper bound proportional to \eqref{eq:closed_form_mse}. 
Therefore, \eqref{eq:closed_form_mse} is adopted as a tractable per-aggregation distortion measure, which serves as the basis for the inference-level objective developed in the next section.

\begin{Remark}[Comparison with Conventional AirComp Error]
\label{rm:difference_from_aircomp}
Unlike conventional AirComp error, where aggregation weights are typically fixed and participating devices are externally given, the layer-wise AirMoE error is coupled with MoE inference. The desired coefficients are token- and layer-dependent gating scores, and the activated devices are determined by both routing and expert placement. As shown in \eqref{eq:closed_form_mse}, the error depends on the transmit powers, the denoising factor, and the channel--power capabilities of the devices hosting the activated experts. This coupling motivates the joint design of AirMoE power control and expert placement.
\end{Remark}

\section{AirMoE Problem Formulation and Decomposition}
\label{sec:airmoe_problem}
Building on the layer-wise AirMoE error, this section calibrates layer sensitivity and defines an inference-level distortion metric for AirMoE. Based on this metric, we formulate the joint expert-placement and power-control problem and show that it can be decomposed without loss of optimality by exploiting the two-timescale structure.

\subsection{Layer Sensitivity of AirMoE}
\label{subsec:from_airmixer_error_to_inference_error}

\begin{figure}[t]
    \centering
    \includegraphics[width=0.7\linewidth]{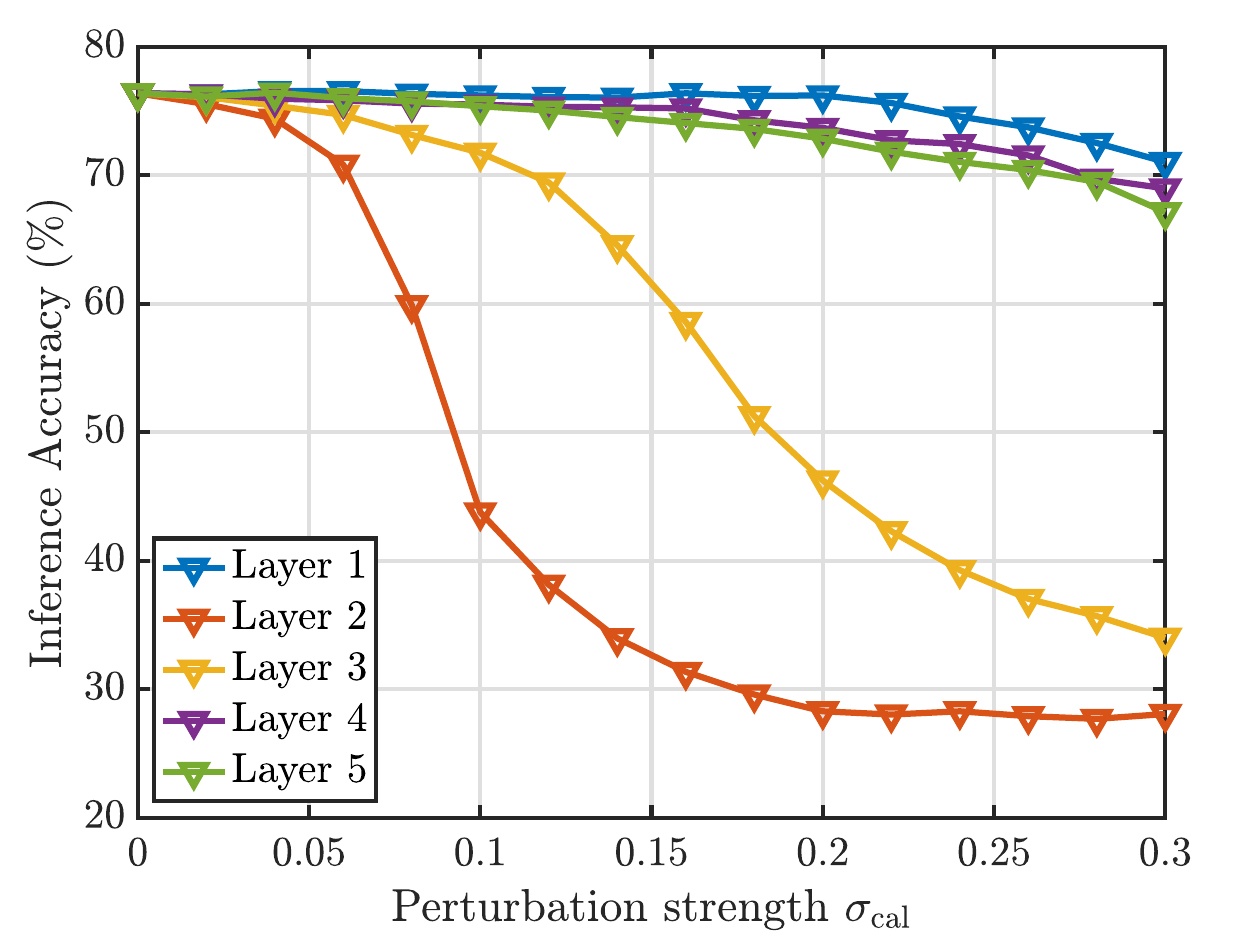}
    \caption{
    Layer-sensitivity calibration for MoE layers of OLMoE model. 
    }
    \label{fig:layer_sensitivity_calibration}
    \vspace{-3mm}
\end{figure}

The layer-wise AirMoE error in \eqref{eq:mse} measures the distortion of a single over-the-air expert aggregation.
For MoE inference, however, E2E performance is affected by the aggregate distortion over multiple tokens and multiple MoE layers. 
Moreover, different layers may exhibit different sensitivities to the same level of aggregation distortion. 
Such layer-dependent sensitivity arises from their distinct roles in representation transformation, heterogeneous routing and gating distributions, and different propagation depths before the perturbations affect the final inference output~\cite{bandarkar2025multilingual,chen2026understanding}.
To capture this inference-level effect, we restore the token and layer indices, where superscript $(t)$ and subscript $\ell$ refer to token $t$ and MoE layer $\ell$, respectively.
We define the overall AirMoE error, denoted by $\mathcal E_{\mathrm A}$, as a layer-sensitivity-weighted long-term aggregation distortion metric, given by
\begin{equation}
\label{eq:high_level_long_term_metric}
\mathcal E_{\mathrm A}
\triangleq
\mathbb{E}
\left[
\sum_{\ell=1}^{L}
a_{\ell}
\mathrm{MSE}_{\ell}^{(t)}
\right]
=
\sum_{\ell=1}^{L}
a_{\ell}
\mathbb{E}
\left[
\mathrm{MSE}_{\ell}^{(t)}
\right],
\end{equation}
where $a_{\ell}\ge 0$ denotes the sensitivity coefficient of layer $\ell$ to aggregation distortion. 
The expectation is taken over input tokens, router-induced expert activations and gating scores, and channel fading realizations. 
$\mathcal E_{\mathrm A}$ provides a tractable objective whose minimization is expected to better preserve over-the-air MoE inference performance.

The layer sensitivities $\{a_\ell\}_{\ell=1}^{L}$ are calibrated offline through controlled perturbation experiments. 
For each layer $\ell$, we evaluate a calibration grid of perturbation strengths
$\mathcal G_{\mathrm{cal}}
=
\left\{
\sigma_{\mathrm{cal},1},
\sigma_{\mathrm{cal},2},
\dots,
\sigma_{\mathrm{cal},J}
\right\}$
and measure the resulting accuracy degradation while keeping all other layers clean. 
This isolates the sensitivity of layer $\ell$ to AirMoE distortion. 
Specifically, for each $\sigma_{\mathrm{cal},j}\in\mathcal G_{\mathrm{cal}}$,  we apply the perturbation only to the normalized aggregation targets $\mathbf{y}_{\ell,\mathrm{tar}}^{(t)}$ of all tokens at layer $\ell$. The resulting perturbed aggregation target, denoted by $\tilde{\mathbf{y}}_{\ell,\mathrm{tar}}^{(t)}(\sigma_{\mathrm{cal},j})$, is given by
\begin{equation}
\label{eq:layer_calib_perturb}
\tilde{\mathbf{y}}_{\ell,\mathrm{tar}}^{(t)}(\sigma_{\mathrm{cal},j})
=
\mathbf{y}_{\ell,\mathrm{tar}}^{(t)}
+
\sigma_{\mathrm{cal},j}
\boldsymbol{\epsilon}_{\ell}^{(t)},
\quad
\boldsymbol{\epsilon}_{\ell}^{(t)}
\sim
\mathcal{N}(\mathbf{0},\mathbf{I}),
\end{equation}
where $\boldsymbol{\epsilon}_{\ell}^{(t)}$ denotes a standard Gaussian perturbation vector. 
The scalar $\sigma_{\mathrm{cal},j}$ controls the perturbation strength, i.e., the standard deviation of the injected Gaussian perturbation in the normalized aggregation domain. 
The perturbed result replaces the clean aggregation result at layer $\ell$ and is propagated through the remaining network for accuracy evaluation.
Fig.~\ref{fig:layer_sensitivity_calibration} illustrates the layer-sensitivity calibration results for representative MoE layers. 
The markedly different accuracy-degradation patterns reveal heterogeneous layer sensitivity to aggregation-domain perturbations, supporting the layer-dependent weighting in \eqref{eq:high_level_long_term_metric}.

Let $\mathrm{Acc}_{\ell}(\sigma_{\mathrm{cal},j})$ denote the task accuracy obtained when perturbation strength $\sigma_{\mathrm{cal},j}$ is applied to layer $\ell$, and let $\mathrm{Acc}_{0}$ denote the clean accuracy.
Given a target relative accuracy drop $\rho$, we assign each layer an unnormalized sensitivity coefficient, denoted by $\tilde{a}_{\ell}$, according to the minimum perturbation strength that causes the prescribed accuracy degradation:
\begin{equation}
\label{eq:unnormalized_al}
\tilde{a}_{\ell}
\triangleq
\left[
\min
\left\{
\sigma_{\mathrm{cal},j}\in\mathcal G_{\mathrm{cal}}:
\mathrm{Acc}_{\ell}(\sigma_{\mathrm{cal},j})
\le
(1-\rho)\mathrm{Acc}_{0}
\right\}
\right]^{-1}.
\end{equation}
A layer whose accuracy reaches the target degradation under a smaller perturbation strength is more sensitive to AirMoE distortion and is therefore assigned a larger coefficient $\tilde{a}_{\ell}$.

To preserve the relative layer sensitivities without changing the average weighting scale across MoE layers, we normalize these coefficients and obtain the layer sensitivity $a_{\ell}$ as
\begin{equation}
\label{eq:al_normalization}
a_{\ell}
=
\frac{
L\tilde{a}_{\ell}
}
{
\sum_{r=1}^{L}
\tilde{a}_{r}
},
\qquad
\frac{1}{L}
\sum_{\ell=1}^{L}
a_{\ell}
=
1.
\end{equation}
Consequently, more distortion-sensitive layers receive larger weights, while the average weight across all MoE layers remains one.

\subsection{Joint Placement and Power-Control Problem}
\label{subsec:joint_problem}

Having defined the overall AirMoE error to characterize inference-relevant aggregation distortion, we now formulate the joint AirMoE design problem for preserving E2E inference performance. The layer-wise AirMoE error in \eqref{eq:closed_form_mse} depends on both the long-term expert-to-device placement and the instantaneous power-control variables. Expert placement is a slow-timescale decision that determines which device hosts each expert and is reused across many inference operations, while power control adjusts the transmit powers and denoising factor for each routing and channel realization. To formulate this joint design, we rewrite the layer-wise AirMoE error in \eqref{eq:closed_form_mse} from the activated-device notation into an expert-indexed form that explicitly depends on placement.

To describe the long-term expert-to-device association, we introduce the binary placement variable $x_{\ell,i,m}$, where $x_{\ell,i,m}=1$ indicates that expert $i$ at layer $\ell$ is placed on device $m$, and $x_{\ell,i,m}=0$ otherwise. We collect all placement variables as
\begin{equation}
\label{eq:placement_constraint_binary}
\mathbf X
\triangleq
\left\{
x_{\ell,i,m}\in\{0,1\}
:
\ell\in\mathcal L,\
i\in\mathcal I,\
m\in\mathcal M
\right\}.
\end{equation}
Under the one-expert-per-device deployment, each expert is assigned to exactly one device:
\begin{equation}
\label{eq:placement_constraint_expert}
\sum_{m\in\mathcal{M}}
x_{\ell,i,m}=1,
\qquad
\forall \ell\in\mathcal{L},\ i\in\mathcal{I},
\end{equation}
and each device hosts at most one expert:
\begin{equation}
\label{eq:placement_constraint_device}
\sum_{\ell\in\mathcal{L}}
\sum_{i\in\mathcal{I}}
x_{\ell,i,m}
\le
1,
\qquad
\forall m\in\mathcal{M}.
\end{equation}

For token $t$ at layer $\ell$, let $S_{\ell}^{(t)}$ denote the activated expert set and let $g_{\ell,i}^{(t)}$ denote the gating score of expert $i$. Given placement $\mathbf{X}$, the device hosting expert $(\ell,i)$ determines its uplink channel and transmit-power budget. To make this placement dependence explicit, we define the placement-induced channel magnitude of expert $(\ell,i)$, denoted by $H_{\ell,i}^{(t)}(\mathbf X)$, as
\begin{equation}
\label{eq:effective_channel_under_placement}
H_{\ell,i}^{(t)}(\mathbf{X})
\triangleq
\sum_{m\in\mathcal{M}}
x_{\ell,i,m}
\left|h_m^{(t)}\right|,
\end{equation}
and the corresponding power budget $P_{\ell,i}(\mathbf X)$  as
\begin{equation}
\label{eq:effective_power_budget_under_placement}
P_{\ell,i}(\mathbf{X})
\triangleq
\sum_{m\in\mathcal{M}}
x_{\ell,i,m}
P_m .
\end{equation}
We collect the transmit powers of the activated experts into the vector $\mathbf p_{\ell}^{(t)}$, given by
\begin{equation}
\label{eq:power_vector_def}
\mathbf{p}_{\ell}^{(t)}
\triangleq
\left\{
p_{\ell,i}^{(t)}:
i\in S_{\ell}^{(t)}
\right\}.
\end{equation}
The server-side denoising factor is denoted by $\eta_{\ell}^{(t)}$.

With these quantities,  we write the layer-wise AirMoE error as a function of the placement and power-control variables, denoted by
$\mathrm{MSE}_{\ell}^{(t)}\left(\mathbf X,\mathbf p_{\ell}^{(t)},\eta_{\ell}^{(t)}\right)$. It is equivalently obtained from \eqref{eq:closed_form_mse} as
\begin{equation}
\label{eq:mse_under_placement}
\begin{aligned}
&
\mathrm{MSE}_{\ell}^{(t)}
\left(
\mathbf{X},
\mathbf{p}_{\ell}^{(t)},
\eta_{\ell}^{(t)}
\right)=\\&
c_{\ell}^{2}
\left[
\sum_{i\in S_{\ell}^{(t)}}
\left(
\frac{
H_{\ell,i}^{(t)}(\mathbf{X})
\sqrt{p_{\ell,i}^{(t)}}
}
{
\eta_{\ell}^{(t)}
}
-
g_{\ell,i}^{(t)}
\right)^{2}
+
\frac{\sigma^{2}}
{
\left(\eta_{\ell}^{(t)}\right)^{2}
}
\right].
\end{aligned}
\end{equation}
% Equation~\eqref{eq:mse_under_placement} explicitly shows that how placement affects the AirMoE error through device-dependent channels and power budgets.
By substituting \eqref{eq:mse_under_placement} into the AirMoE error in \eqref{eq:high_level_long_term_metric},
the joint placement and power-control problem is formulated as
\begin{equation}
\label{eq:joint_placement_power_problem}
\begin{aligned}
\min_{\mathbf{X},\,\{\mathbf{p}_{\ell}^{(t)},\eta_{\ell}^{(t)}\}}
\quad
&
\sum_{\ell=1}^{L}
a_{\ell}
\mathbb{E}
\left[
\mathrm{MSE}_{\ell}^{(t)}
\left(
\mathbf{X},
\mathbf{p}_{\ell}^{(t)},
\eta_{\ell}^{(t)}
\right)
\right]
\\
\mathrm{s.t.}
\quad
&
0\le p_{\ell,i}^{(t)}
\le
P_{\ell,i}(\mathbf{X}),
\qquad
i\in S_{\ell}^{(t)},\ \forall \ell,t,
\\
&
\eta_{\ell}^{(t)}>0,
\qquad
\forall \ell,t,
\\
&
\eqref{eq:placement_constraint_binary},
\eqref{eq:placement_constraint_expert},
\eqref{eq:placement_constraint_device}.
\end{aligned}
\end{equation}
Problem~\eqref{eq:joint_placement_power_problem} is a mixed-integer stochastic optimization problem with a two-timescale decision structure. The placement variable $\mathbf X$ is a long-term binary decision fixed before online inference, whereas $\mathbf p_{\ell}^{(t)}$ and $\eta_{\ell}^{(t)}$ are continuous recourse variables adapted to each routing and channel realization. Directly solving this problem is challenging due to the nested coupling between the combinatorial expert-placement decision and the scenario-dependent power-control policy under the long-term expectation. To make the problem tractable, we exploit the two-timescale structure and develop an equivalent decomposition in the next subsection.

\subsection{Two-Timescale Decomposition}
\label{subsec:exact_decomposition}

We next decompose problem~\eqref{eq:joint_placement_power_problem} by exploiting its two-timescale structure. The key observation is that, for any fixed placement $\mathbf X$, the power-control variables are adapted to each instantaneous routing and channel realization. Since both the layer-wise AirMoE error and the feasible set are local to each layer-token aggregation, with no coupling across tokens, layers, or channel realizations, the online power-control variables can be optimized pointwise for each instantaneous AirMoE aggregation.

For a given placement $\mathbf X$, the corresponding per-aggregation power-control problem is formulated as
\begin{equation}
\label{eq:min_mse}
\begin{aligned}
\min_{\mathbf p_{\ell}^{(t)},\eta_{\ell}^{(t)}}
\quad
&
\mathrm{MSE}_{\ell}^{(t)}
\left(
\mathbf X,
\mathbf p_{\ell}^{(t)},
\eta_{\ell}^{(t)}
\right)
\\
\mathrm{s.t.}
\quad
&
0\le p_{\ell,i}^{(t)}
\le
P_{\ell,i}(\mathbf X),
\qquad
i\in S_{\ell}^{(t)},
\\
&
\eta_{\ell}^{(t)}>0.
\end{aligned}
\end{equation}
The minimization in \eqref{eq:min_mse} is conditioned on the current activated expert set, gating scores, and channel realization. Its optimal objective value is denoted by $\mathrm{MSE}_{\ell}^{(t),\star}(\mathbf X)$, whose dependence on $\mathbf X$ arises through the placement-induced channel magnitudes and power budgets.

Substituting the optimal value in \eqref{eq:min_mse} into \eqref{eq:joint_placement_power_problem} yields the following long-timescale expert-placement problem:
\begin{equation}
\label{eq:min_placement}
\begin{aligned}
\min_{\mathbf X}
\quad
&
\sum_{\ell=1}^{L}
a_{\ell}
\mathbb E
\left[
\mathrm{MSE}_{\ell}^{(t),\star}(\mathbf X)
\right]\\
\mathrm{s.t.}
\quad
&
\eqref{eq:placement_constraint_binary},
\eqref{eq:placement_constraint_expert},
\eqref{eq:placement_constraint_device}.
\end{aligned}
\end{equation}

\begin{Proposition}[Optimality of Two-Timescale Decomposition]
\label{prop:problem_decomposition}
Problem~\eqref{eq:joint_placement_power_problem} has the same optimal objective value as the nested optimization consisting of the inner power-control problem in \eqref{eq:min_mse} and the outer placement problem in \eqref{eq:min_placement}. 
\end{Proposition}

\noindent\textit{Proof:} See Appendix~\ref{subsec:problem_decomposition}.

Proposition~\ref{prop:problem_decomposition} reduces the AirMoE design to two timescale-separated subproblems. The per-aggregation power-control problem in \eqref{eq:min_mse} is first solved in Section~\ref{sec:power_control_minimizing_airmixer_error}. Then, based on the resulting optimal aggregation distortion, the expert-placement problem in \eqref{eq:min_placement} is addressed in Section~\ref{sec:joint_expert_placement_power_control}.

\section{Power Control for AirMoE}
\label{sec:power_control_minimizing_airmixer_error}

In this section, we study the short-timescale power-control problem in \eqref{eq:min_mse} by jointly optimizing the device transmit powers and the server-side denoising factor. We then characterize the structure of the optimal policy and its behavior under different SNR regimes.

\subsection{Power-Control Problem}
\label{subsec:power_control_problem}
With the expert placement fixed, each activated expert is uniquely associated with an activated device.
For notational clarity, we therefore rewrite the expert-indexed formulation in Section~\ref{sec:airmoe_problem} using the activated-device notation introduced in Section~\ref{sec:overview_airmoe_inference}.
Consider one AirMoE aggregation at layer $\ell$ over the activated-device set $\mathcal A$, with gating scores $\{g_m\}$, channel coefficients $\{h_m\}$, and power budgets $\{P_m\}$ for $m\in\mathcal A$.
The inner problem in \eqref{eq:min_mse} then reduces to optimizing the transmit powers $\{p_m\}$ and the denoising factor $\eta$.
Since the layer-dependent factor $c_\ell^2$ in \eqref{eq:closed_form_mse} is fixed for the considered aggregation, it does not affect the optimizer and is omitted in this section.
Thus, the power-control problem in \eqref{eq:min_mse} can be equivalently written as
\begin{equation} \label{eq:pc_problem_original} \begin{aligned} 
\min_{\{p_m\},\,\eta} \quad & \sum_{m\in\mathcal A} \left( \frac{|h_m|\sqrt{p_m}}{\eta} - g_m \right)^2 + \frac{\sigma^2}{\eta^2} \\ \mathrm{s.t.} \quad & 0\le p_m\le P_m, \qquad m\in\mathcal A \\&\eta>0 . \end{aligned} 
\end{equation}
The objective in \eqref{eq:pc_problem_original} reveals the roles of the two control variables $\{p_m\}$ and $\eta$. The transmit powers $\{p_m\}$ tune the effective aggregation coefficients $|h_m|\sqrt{p_m}/\eta$ toward the desired gating scores $\{g_m\}$ under individual power budgets.
The denoising factor $\eta$ governs the tradeoff between coefficient matching and noise suppression:
increasing $\eta$ suppresses the effective noise term $\sigma^2/|\eta|^2$ but also shrinks the achievable coefficient range, which may cause residual mismatch for devices with limited channel--power capability.

Problem~\eqref{eq:pc_problem_original} can be solved by exact variable elimination. For any fixed denoising factor $\eta$, the noise term is constant, and the objective is separable with respect to the transmit powers ${p_m}$ across activated devices. Hence, the optimal transmit powers can be first derived as functions of $\eta$. Substituting these optimal powers back into the objective yields a one-dimensional value function of $\eta$, whose minimization gives the global optimum without loss of optimality.
\begin{figure}[t]
    \centering
    \includegraphics[width=0.7\linewidth]{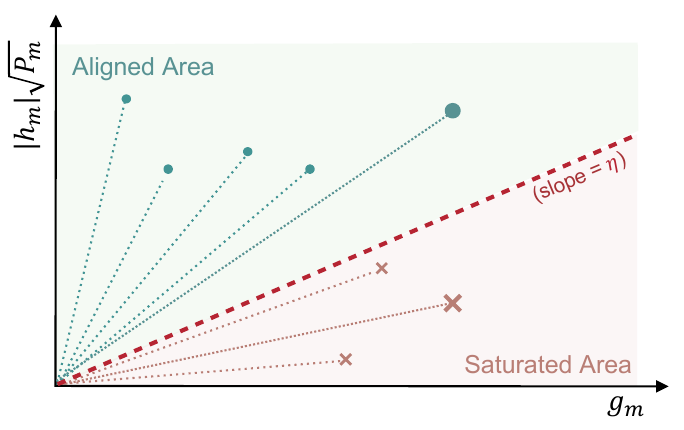}
    \caption{Geometric interpretation of the saturation threshold.}
    \label{fig:saturation_threshold_geometry}
    \vspace{-3mm}
\end{figure}

\subsection{Optimal Power-Control Policy}
\label{subsec:optimal_power_control_policy}

We now derive the optimal AirMoE power-control policy following the variable-elimination procedure described above.
For a given denoising factor $\eta>0$, the effective-noise term in \eqref{eq:pc_problem_original} is fixed, and the remaining coefficient-mismatch terms are separable across activated devices.
Therefore, the transmit-power optimization reduces to the following independent per-device subproblem:
\begin{equation}
\label{eq:pc_subproblem_pm}
\begin{aligned}
\min_{p_m} \quad & \left( \frac{|h_m|\sqrt{p_m}}{\eta} - g_m \right)^2 \\ \mathrm{s.t.} \quad & 0\le p_m\le P_m, \qquad m\in\mathcal A .
\end{aligned}
\end{equation}
The following lemma gives the closed-form solution of this subproblem.

\begin{Lemma}[Optimal transmit power for fixed $\eta$]
\label{lem:pm_opt}
For any fixed $\eta>0$, the conditional optimal transmit power of device $m\in\mathcal A$ is
\begin{equation}
\label{eq:pm_opt_eta}
\tilde{p}_m(\eta)
=
\min
\left\{
\left(
\frac{\eta g_m}{|h_m|}
\right)^2,
P_m
\right\}.
\end{equation}
\end{Lemma}

\noindent\textit{Proof:}
For fixed $\eta$, the objective in \eqref{eq:pc_subproblem_pm} is a non-negative quadratic mismatch term, whose zero-mismatch point is
$p_m=\left( \frac{\eta g_m}{|h_m|} \right)^2$.
Projecting this point onto the feasible interval $[0,P_m]$ yields \eqref{eq:pm_opt_eta}.

Lemma~\ref{lem:pm_opt} shows that, for a fixed denoising factor $\eta$, each activated device operates in one of two regimes. If the power required for exact coefficient matching does not exceed the power budget, i.e., $ \left( \frac{\eta g_m}{|h_m|} \right)^2 \le P_m$, device $m$ can align its effective coefficient with the gating score $g_m$ and eliminate the coefficient mismatch. Otherwise, device $m$ is power-limited and transmits at full power $p_m=P_m$. In this case, it remains mismatched and incurs the residual error  $\left( g_m - \frac{|h_m|\sqrt{P_m}}{\eta} \right)^2$, which is referred to as the saturated regime.

To clarify this two-regime structure, we provide a geometric interpretation in the 
$(g_m,|h_m|\sqrt{P_m})$ plane, as illustrated in Fig.~\ref{fig:saturation_threshold_geometry}.
For a fixed denoising factor $\eta$, the line
$|h_m|\sqrt{P_m}=\eta g_m$
has slope $\eta$ and separates the aligned and saturated regimes.
Each device $m$ is represented by the point $(g_m,|h_m|\sqrt{P_m})$, whose slope from the origin defines its saturation threshold $\tau_m$, given by
\begin{equation}
\label{eq:tau_def}
\tau_m
\triangleq
\frac{|h_m|\sqrt{P_m}}{g_m},
\qquad m\in\mathcal A .
\end{equation}
Accordingly, device $m$ remains aligned as long as $\eta\le\tau_m$, and becomes saturated once $\eta$ exceeds $\tau_m$.

Having obtained the optimal transmit powers for any fixed $\eta$, we now turn to the outer optimization of the denoising factor.
 By substituting the optimal transmit powers in Lemma~\ref{lem:pm_opt} into the layer-wise AirMoE error, we obtain the reduced objective as a function of $\eta$, denoted by $\mathrm{MSE}(\eta)$. The resulting one-dimensional optimization problem is given by
\begin{equation}
\label{eq:eta_problem}
\min_{\eta>0}
\quad
\mathrm{MSE}(\eta)
=
\sum_{m\in\mathcal A}
\max
\left\{
0,\,
g_m-\frac{|h_m|\sqrt{P_m}}{\eta}
\right\}^{2}
+
\frac{\sigma^2}{\eta^{2}} .
\end{equation}
The main difficulty in solving \eqref{eq:eta_problem} lies in the max operator, which makes the objective piecewise smooth with an $\eta$-dependent active set. For each device, the max term is active only when the device is saturated, i.e., when $\eta>\tau_m$. Thus, as $\eta$ varies over $(0,+\infty)$, each threshold $\tau_m$ acts as a breakpoint at which the saturation status of device $m$ changes. To remove the max operator, we sort the saturation thresholds as
\begin{equation}
\label{eq:tau_order}
\tau_{(1)}
\le
\tau_{(2)}
\le
\cdots
\le
\tau_{(|\mathcal A|)},
\end{equation}
where the parenthesized index denotes the ordering induced by the sorted thresholds. 
With $\tau_{(0)}\triangleq 0$ and $\tau_{(|\mathcal A|+1)}\triangleq+\infty$, the feasible domain of
$\eta>0$ is covered by the following regions:
\begin{equation}
\label{eq:eta_region_def}
\mathcal R_n
\triangleq
[\tau_{(n)},\,\tau_{(n+1)}],
\qquad
n=0,1,\dots,|\mathcal A|.
\end{equation}

Within the interior of $\mathcal R_n$, the saturated-device set remains unchanged. 
Specifically, for any $\eta\in\mathcal R_n$, the devices indexed by $(1),\dots,(n)$ are saturated, whereas devices $(n+1),\dots,(|\mathcal A|)$ remain aligned. Therefore, solving \eqref{eq:eta_problem} is equivalent to solving $|\mathcal A|+1$ smooth subproblems over these regions and selecting the one with the smallest objective value. Let $\mathrm{MSE}_n(\eta)$ denote the layer-wise AirMoE error for $\eta\in\mathcal R_n$.
Since only the first $n$ devices can incur coefficient mismatch, $\mathrm{MSE}_n(\eta)$ is given by
\begin{equation}
\label{eq:mse_interval}
\mathrm{MSE}_n(\eta)
=
\sum_{j=1}^{n}
\left(
g_{(j)}
-
\frac{|h_{(j)}|\sqrt{P_{(j)}}}{\eta}
\right)^2
+
\frac{\sigma^2}{\eta^2},
\qquad
\eta\in\mathcal R_n .
\end{equation}

For $n=0$, the summation in \eqref{eq:mse_interval} is empty because no device is saturated over $\mathcal R_0=(0,\tau_{(1)}]$. 
Thus, $\mathrm{MSE}_0(\eta)=\sigma^2/\eta^2$, which is strictly decreasing in $\eta$. 
Its minimum over $\mathcal R_0$ is attained at the boundary $\eta=\tau_{(1)}$, which is covered by the projection from the neighboring region. 
The remaining regions admit the closed-form minimizer in the following lemma.

\begin{Lemma}[Region-wise optimal denoising factor]
\label{lem:eta_region_opt}
For any region $\mathcal R_n$ with $1\le n\le |\mathcal A|$, $\mathrm{MSE}_n(\eta)$ has a unique stationary point over $\eta>0$, given by
\begin{equation}
\label{eq:eta_n_stationary}
\hat{\eta}_n^\star
=
\frac{
\sum_{j=1}^{n} |h_{(j)}|^2 P_{(j)}+\sigma^2
}
{
\sum_{j=1}^{n}
|h_{(j)}|\sqrt{P_{(j)}}g_{(j)}
}.
\end{equation}
The minimizer of $\mathrm{MSE}_n(\eta)$ over $\mathcal R_n$ is obtained by projecting $\hat{\eta}_n^\star$ onto $\mathcal R_n$:
\begin{equation}
\label{eq:eta_n_proj}
\eta_n^\star
=
\min
\left\{
\max
\left\{
\hat{\eta}_n^\star,\,
\tau_{(n)}
\right\},
\tau_{(n+1)}
\right\}.
\end{equation}
\end{Lemma}

\noindent\textit{Proof:} See Appendix~\ref{subsec:eta_region_opt}.

Given the region-wise minimizers in Lemma~\ref{lem:eta_region_opt}, the globally optimal denoising factor is obtained by selecting the candidate
that yields the minimum layer-wise AirMoE error. Specifically, the globally optimal region index is given by
\begin{equation}
\label{eq:n_global_def}
n^\star
\in
\arg\min_{1\le n\le |\mathcal A|}
\mathrm{MSE}(\eta_n^\star).
\end{equation}
The complete globally optimal AirMoE power-control policy is characterized
in the following theorem.

\begin{Theorem}[Optimal AirMoE power control]
\label{thm:global_policy_fixedA}
For the AirMoE power-control problem in \eqref{eq:pc_problem_original}, the globally optimal denoising factor  $\eta^\star$ is given by
\begin{equation}
\label{eq:eta_global_def}
\eta^\star=\eta_{n^\star}^\star,
\end{equation}
where $n^\star$ is defined in \eqref{eq:n_global_def}.
Given $\eta^\star$, the globally optimal transmit power for device $m\in\mathcal A$ is given by
\begin{equation}
\label{eq:global_policy_compact}
p_m^\star
=
\min
\left\{
\left(
\frac{\eta^{\star} g_m}{|h_m|}
\right)^2,
P_m
\right\},
\qquad
m\in\mathcal A .
\end{equation}
\end{Theorem}

The threshold structure above further provides guidance for the subsequent expert-placement design.

\subsection{Asymptotic Analysis}
\label{subsec:asymptotic_analysis_extreme_snr}

We further examine the optimal AirMoE power-control policy in Theorem~\ref{thm:global_policy_fixedA} under two extreme SNR regimes.

In the high-SNR regime, i.e., $\sigma^2\to 0$, the layer-wise AirMoE error in \eqref{eq:eta_problem} is dominated by the coefficient-mismatch term.
The optimal policy therefore enforces exact coefficient alignment for all activated devices, which requires
$\eta
\le
\tau_{(1)}
=
\min_{m\in\mathcal A}
\frac{|h_m|\sqrt{P_m}}{g_m}$.
Among all zero-mismatch choices in $(0,\tau_{(1)}]$, $\eta=\tau_{(1)}$ minimizes the denoised noise term.
Hence, in the high-SNR limit, the optimal denoising factor is 
\begin{equation}
\label{eq:eta_high_snr_star}
\eta^\star=\tau_{(1)}.
\end{equation}
The corresponding optimal transmit power is
\begin{equation}
\label{eq:p_high_snr}
p_m^\star
=
\left(\frac{\tau_{(1)} g_m}{|h_m|}\right)^2,
\quad
m\in\mathcal A,
\end{equation}
which realizes channel inversion to match the desired gating scores.

In the low-SNR regime, i.e., $\sigma^2\to\infty$, the denoised receiver-noise term dominates \eqref{eq:eta_problem}.
The optimal policy therefore selects a large denoising factor, placing the solution in the fully saturated region where all activated devices transmit at full power:
\begin{equation}
\label{eq:p_low_snr}
p_m^\star
=P_m,
\quad
m\in\mathcal A.
\end{equation}
The corresponding denoising factor is the stationary point of the last region:
\begin{equation}
\label{eq:eta_low_snr}
\eta^\star
=
\hat{\eta}_{|\mathcal A|}^\star
=
\frac{
\sum_{m\in\mathcal A}
|h_m|^2P_m
+\sigma^2
}
{
\sum_{m\in\mathcal A}
|h_m|\sqrt{P_m}g_m
}.
\end{equation}

Therefore, the closed-form AirMoE policy continuously connects channel-inversion-based coefficient matching at high SNR and full-power transmission at low SNR, as further validated by simulations.

\section{Expert Placement for AirMoE}
\label{sec:joint_expert_placement_power_control}

In this section, we study the long-timescale expert-placement problem for AirMoE.
Building on the optimized per-aggregation power control derived in the preceding section, we develop an activation- and channel-aware placement rule and then discuss its extension to the multi-expert-per-device setting.

\subsection{Activation- and Channel-Aware Expert Placement}
\label{subsec:rank_preserving_placement_rule}

The preceding section solves the short-timescale power-control problem for each instantaneous AirMoE aggregation. We now return to the expert-indexed placement notation and study the long-timescale expert-placement problem in \eqref{eq:min_placement}. For a given placement $\mathbf X$, applying the optimal power-control policy in Theorem~\ref{thm:global_policy_fixedA} attains the optimal per-aggregation distortion $\mathrm{MSE}_{\ell}^{(t),\star}(\mathbf X)$, which constitutes the objective term in \eqref{eq:min_placement}. However, due to its dependence on routing realizations, fading states, and saturation regimes, $\mathrm{MSE}_{\ell}^{(t),\star}(\mathbf X)$ does not admit a tractable closed-form expression with respect to $\mathbf X$. To obtain an interpretable and tractable placement rule, we upper bound the layer-weighted long-term AirMoE error by evaluating the aggregation error under a feasible denoising factor, as provided in Lemma~\ref{lem:placement_layer_weighted_bound}.

\begin{Lemma}[AirMoE error upper bound]
\label{lem:placement_layer_weighted_bound}
For any placement $\mathbf X$, the layer-weighted long-term distortion under optimal  power control satisfies
\begin{equation}
\label{eq:placement_factorized_bound_raw}
\begin{aligned}
&
\sum_{\ell=1}^{L}
a_\ell
\mathbb E
\left[
\mathrm{MSE}_{\ell}^{(t),\star}(\mathbf X)
\right]
\le
\sigma^2
\sum_{\ell=1}^{L}
a_\ell c_\ell^2\\
&\qquad \qquad\times
\mathbb E
\left[
\sum_{i\in S_\ell^{(t)}}
\sum_{m\in\mathcal M}
x_{\ell,i,m}
\frac{
\left(g_{\ell,i}^{(t)}\right)^2
}
{
|h_m^{(t)}|^2P_m
}
\right].
\end{aligned}
\end{equation}
\end{Lemma}

\noindent\textit{Proof:} See Appendix~\ref{subsec:placement_layer_weighted_bound_proof}.

Lemma~\ref{lem:placement_layer_weighted_bound} provides a tractable upper bound on the placement objective in \eqref{eq:min_placement}. To reveal its separable structure, we first rewrite the summation over the activated expert set using an activation indicator. By the linearity of expectation and the fact that the placement variable $\mathbf X$ is fixed before inference, the right-hand side of \eqref{eq:placement_factorized_bound_raw} can be rewritten as
\begin{equation}
\label{eq:placement_upper_bound_factorized}
\begin{aligned}
&
\sigma^2
\sum_{\ell=1}^{L}
a_\ell c_\ell^2
\mathbb E
\left[
\sum_{i\in S_\ell^{(t)}}
\sum_{m\in\mathcal M}
x_{\ell,i,m}
\frac{
\left(g_{\ell,i}^{(t)}\right)^2
}
{
|h_m^{(t)}|^2P_m
}
\right]
\\
&=
\sigma^2
\sum_{\ell=1}^{L}
\sum_{i\in\mathcal I}
\sum_{m\in\mathcal M}
x_{\ell,i,m}
a_\ell c_\ell^2
\mathbb E
\left[
\mathbf 1\{i\in S_\ell^{(t)}\}
\frac{
\left(g_{\ell,i}^{(t)}\right)^2
}
{
|h_m^{(t)}|^2P_m
}
\right]
\\
&=
\sigma^2
\sum_{\ell=1}^{L}
\sum_{i\in\mathcal I}
\sum_{m\in\mathcal M}
x_{\ell,i,m}\\
&\qquad \qquad \times
\underbrace{
a_\ell c_\ell^2
\mathbb E
\left[
\mathbf 1\{i\in S_\ell^{(t)}\}
\left(g_{\ell,i}^{(t)}\right)^2
\right]
}_{W_{\ell,i}}
\underbrace{
\mathbb E
\left[
\frac{1}{|h_m^{(t)}|^2P_m}
\right]
}_{R_m}.
\end{aligned}
\end{equation}
where the last equality follows because the router-induced activation and gating scores are independent of wireless fading. The term $W_{\ell,i}$ represents the expert-side importance, while $R_m$ represents the device-side cost. Specifically, the expert-side importance is given by
\begin{equation}
\begin{aligned}
\label{eq:placement_expert_weight}
W_{\ell,i}
&\triangleq
a_\ell c_\ell^2
\mathbb E
\left[
\mathbf 1\{i\in S_\ell^{(t)}\}
\left(g_{\ell,i}^{(t)}\right)^2
\right]\\
&=
a_\ell c_\ell^2
\Pr(i\in S_\ell^{(t)})
\mathbb E
\left[
\left(g_{\ell,i}^{(t)}\right)^2
\,\middle|\,
i\in S_\ell^{(t)}
\right].
\end{aligned}
\end{equation}
Here, $a_\ell$ accounts for the layer sensitivity and $c_\ell^2$ captures the layer-wise expert-output scale.
The activation probability $\Pr(i\in S_\ell^{(t)})$ reflects how frequently expert $(\ell,i)$ is selected by the router, while the conditional expectation reflects how strongly it contributes once selected.
Thus, a larger $W_{\ell,i}$ indicates that expert $(\ell,i)$ is more critical to AirMoE inference reliability, in the sense that its aggregation distortion has a larger long-term impact on the final inference performance.

The corresponding device-side cost in
\eqref{eq:placement_upper_bound_factorized} is given by
\begin{equation}
\label{eq:placement_device_cost_exact}
R_m
\triangleq
\mathbb E
\left[
\frac{1}{|h_m^{(t)}|^2P_m}
\right].
\end{equation}
This term characterizes the long-term inverse channel--power capability of device $m$, where a smaller $R_m$ indicates a stronger device for AirMoE uplink aggregation. 

Since expert placement is configured at a slow timescale, both $W_{\ell,i}$ and $R_m$ can be estimated before online inference from long-term statistics. The expert-side importance $W_{\ell,i}$ can be obtained through offline model profiling over representative samples, while the device-side cost $R_m$ can be estimated from long-term channel and power statistics. To avoid numerical instability caused by rare deep-fading realizations, channel samples with power below a prescribed threshold are excluded from the estimation.

Using the expert-side importance in \eqref{eq:placement_expert_weight} and the device-side cost in \eqref{eq:placement_device_cost_exact}, the long-term placement design in \eqref{eq:joint_placement_power_problem} reduces to the following separable assignment problem:

\begin{equation}
\label{eq:placement_separable_problem}
\begin{aligned}
\min_{\mathbf X}
\quad
&
\sum_{\ell=1}^{L}
\sum_{i\in\mathcal I}
\sum_{m\in\mathcal M}
x_{\ell,i,m}
W_{\ell,i}
 R_m
\\
\mathrm{s.t.}
\quad
&\eqref{eq:placement_constraint_binary},\eqref{eq:placement_constraint_expert},\eqref{eq:placement_constraint_device}.
\end{aligned}
\end{equation}
This multiplicatively separable structure leads to a rank-preserving placement rule, as established in Theorem~\ref{thm:placement_rank_preserving}.

\begin{Theorem}[Rank-preserving optimal placement]
\label{thm:placement_rank_preserving}
Let $\mathcal J=\mathcal L\times\mathcal I$ be the set of expert indices with $|\mathcal J|=LI$, and let $\mathcal M$ be the set of edge devices with $|\mathcal M|\ge |\mathcal J|$.
For each $j=(\ell,i)\in\mathcal J$, define $W_j=W_{\ell,i}$.
Relabel the experts and devices such that
\[
W_1\ge W_2\ge \cdots \ge W_{|\mathcal J|},
\qquad
 R_1\le  R_2\le \cdots \le  R_{|\mathcal M|}.
\]
Then an optimal solution to \eqref{eq:placement_separable_problem} assigns the $j$-th largest-weight expert to the $j$-th lowest-cost device, i.e.,
\begin{equation}
\label{eq:rank_preserving_assignment}
x_{j,m}^{\star}
=
\begin{cases}
1, & m=j,\quad j=1,\dots,|\mathcal J|,\\
0, & \text{otherwise},
\end{cases}
\end{equation}
where only the first $|\mathcal J|$ devices are used.
\end{Theorem}

\noindent\textit{Proof:} See Appendix~\ref{subsec:rank_preserving_placement_proof}.

Theorem~\ref{thm:placement_rank_preserving} shows that the separable assignment problem admits a rank-preserving solution: experts with larger layer sensitivity, aggregation-output scale, activation frequency, or gating-score magnitude should be assigned to devices with stronger long-term channel--power capability. This result provides an inference-aware placement principle: wireless resources should be prioritized for experts whose aggregation errors have larger long-term impact on inference performance.

\subsection{Extension to Multi-Expert Devices}
\label{subsec:extension_multi_expert_devices}

The rank-preserving rule in Theorem~\ref{thm:placement_rank_preserving} is derived under the one-expert-per-device architecture. In practical deployments, however, a device may have sufficient storage to host multiple experts. In this case, multi-expert hosting can be allowed across different MoE layers, while experts within the same layer should still be placed on distinct devices to preserve simultaneous AirMoE uplink aggregation. 

Let $C_m$ denote the storage capacity of device $m$, measured by the maximum number of expert instances it can host. The one-expert-per-device constraint can be relaxed to
\begin{equation}
\label{eq:many_to_one_capacity_constraint}
\sum_{\ell\in\mathcal L}
\sum_{i\in\mathcal I}
x_{\ell,i,m}
\le
C_m,
\qquad
\forall m\in\mathcal M,
\end{equation}
while imposing the layer-wise non-overlapping constraint
\begin{equation}
\label{eq:many_to_one_layer_constraint}
\sum_{i\in\mathcal I}
x_{\ell,i,m}
\le
1,
\qquad
\forall \ell\in\mathcal L,\ m\in\mathcal M .
\end{equation}
The expert-assignment constraint in \eqref{eq:placement_constraint_expert} remains unchanged.

With \eqref{eq:many_to_one_capacity_constraint} and \eqref{eq:many_to_one_layer_constraint}, the placement problem becomes a capacitated binary assignment problem with layer-wise conflict constraints. The closed-form rank-preserving solution in Theorem~\ref{thm:placement_rank_preserving} no longer directly applies, since the allocation of a device to one expert may restrict the placement choices of other experts in the same layer. Nevertheless, the theorem still provides a useful ordering principle: experts with larger importance weights $W_{\ell,i}$ should be preferentially assigned to devices with smaller channel--power costs $R_m$, subject to the storage-capacity and layer-wise non-overlapping constraints. Since expert placement is configured offline, the resulting problem can be solved using standard integer linear programming techniques~\cite{wolsey1998integer}.

\section{Experimental Results}
\label{sec:experimental_results}
\begin{figure*}[t]
    \centering
    \makebox[0.5\textwidth][c]{%
        \subfloat[Power-control ablation\label{fig:arceasy_powercontrol_acc}]{
            \includegraphics[width=0.43\textwidth]{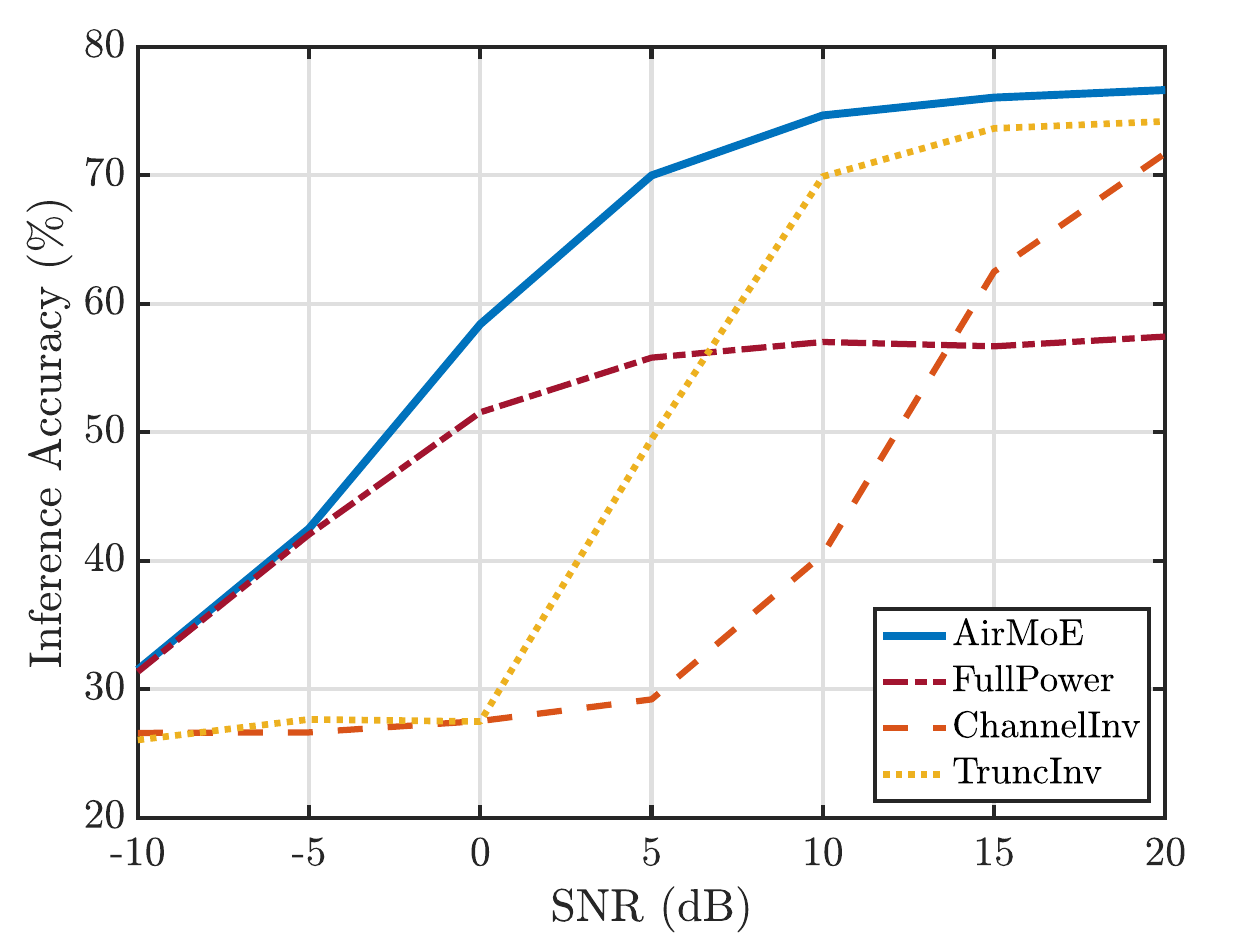}
        }
    }%
    \makebox[0.5\textwidth][c]{%
        \subfloat[Expert-placement ablation\label{fig:arceasy_placement_acc}]{
            \includegraphics[width=0.43\textwidth]{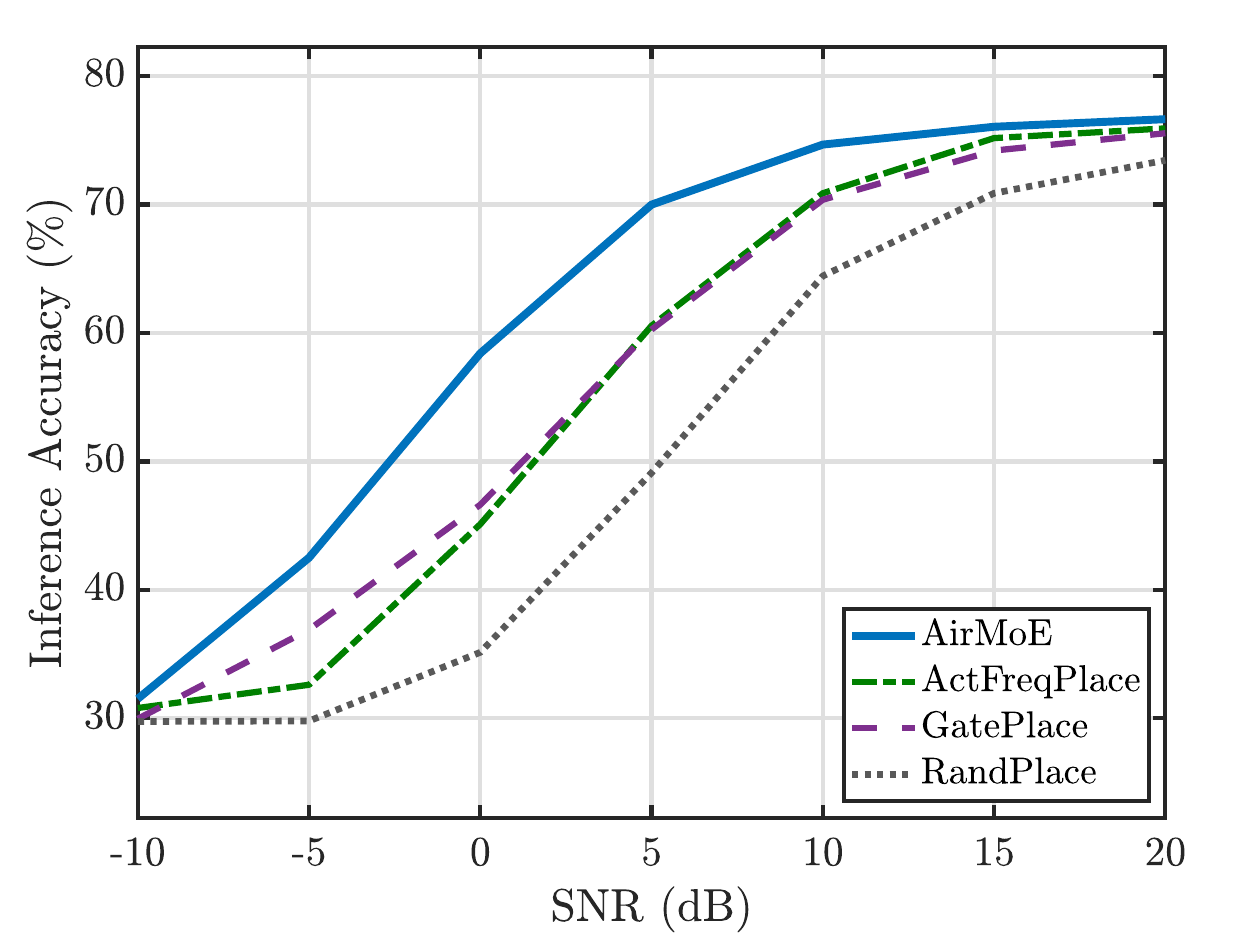}
        }
    }
    \caption{Inference accuracy of AirMoE and its power-control and expert-placement ablations under different SNRs.}
    \label{fig:snr_performance}
    \vspace{-3mm}
\end{figure*}

\subsection{Experimental Setup}
\label{subsec:experimental_setup}

\subsubsection{MoE Configuration}
We evaluate the proposed AirMoE framework using the OLMoE-1B-7B-0924 model~\cite{muennighoff2025olmoe}. The model contains $L=16$ MoE layers, each with $64$ experts, and adopts Top-$K$ routing with $K=8$ activated experts per token. Experiments are conducted on ARC-Easy~\cite{clark2018arc}, a standard English multiple-choice question-answering benchmark, using the lm-evaluation-harness toolkit~\cite{lmevalharness}.
The offline statistics used by AirMoE are estimated from a small calibration set before inference. These include the normalization statistics $\{\boldsymbol{\mu}_{\ell},c_\ell\}$, the layer-sensitivity coefficients $\{a_\ell\}$ calibrated using a target relative accuracy drop of $\rho=3\%$, and the expert-side weights $\{W_{\ell,i}\}$ computed according to \eqref{eq:placement_expert_weight}.

\subsubsection{Wireless Setup}
For the uplink channels, the instantaneous channel coefficient of device $m$ is modeled as
$h_m^{(t)}\sim\mathcal{CN}(0,\Omega_m)$, where $\Omega_m$ denotes the long-term average channel power. Following the standard large-scale fading model~\cite{goldsmith2005wireless}, we generate $\Omega_m \propto d_m^{-2.8}10^{X_m/10}$,
where $d_m$ is the distance between device $m$ and the edge server, uniformly sampled from $[30,120]$ m, and $X_m\sim\mathcal{N}(0,\sigma_{\rm sh}^{2})$ denotes log-normal shadowing in dB. The shadowing standard deviation $\sigma_{\rm sh}$ controls the degree of long-term channel heterogeneity across devices and is set to $4~\mathrm{dB}$ unless otherwise specified.
The network-average receive SNR is defined as
\begin{equation}
\label{eq:snr_def}
\mathrm{SNR}
\triangleq
10\log_{10}
\left(
\frac{1}{M}
\sum_{m=1}^{M}
\frac{P_m\Omega_m}{\sigma^2}
\right)
~\mathrm{dB},
\end{equation}
where $P_m$ is the transmit-power budget of device $m$, and $\sigma^2$ is the effective real-domain receiver-noise variance. We set $P_m=0.2~\mathrm{W}$ for all devices to isolate long-term channel heterogeneity, and vary $\sigma^2$ to obtain different SNR values. The device-side costs $\{R_m\}$ used for expert placement are estimated from the corresponding long-term channel statistics.

\subsubsection{Benchmarking Schemes}
We compare the proposed AirMoE scheme, which integrates  optimal power control and layer-aware expert placement, with two groups of ablation baselines to separately assess the contributions of power control and expert placement. All baselines are implemented under the same over-the-air MoE aggregation framework and differ only in the specified design component.

The power-control baselines keep the proposed layer-aware placement fixed and differ only in the instantaneous power-control policy.

\begin{itemize}
    \item \emph{Full-Power Transmission} (FullPower): 
    All activated devices transmit at their maximum power budgets, and the server-side denoising factor is optimized for the resulting over-the-air aggregation.

\item \emph{Channel Inversion} (ChannelInv): 
The denoising factor is determined by the minimum saturation threshold among activated devices,
$\eta_{\ell}^{(t)}
    =
    \min_{m\in\mathcal A_\ell^{(t)}}
    \frac{|h_m^{(t)}|\sqrt{P_m}}{g_{\ell,m}^{(t)}}$,
which ensures exact coefficient alignment for all activated experts.

    \item \emph{Truncated Channel Inversion} (TruncInv):
    Activated devices with weak instantaneous channels, i.e., $|h_m^{(t)}|<\xi$, are excluded from transmission.
    Channel-inversion-based aggregation is then performed over the remaining activated devices. Unless otherwise specified, we set $\xi=0.2$.

\end{itemize}

The expert placement baselines use the same optimal power-control policy as AirMoE, while adopting different expert-ranking criteria for assignment.

\begin{itemize}
\item \emph{Activation-Frequency-Based Placement} (ActFreqPlace): 
    Experts are ranked only by the activation-frequency component in \eqref{eq:placement_expert_weight}, i.e., $\Pr(i\in S_\ell^{(t)})$.
    
    \item \emph{Gating-Based Placement} (GatePlace):
    Experts are ranked only by the activation-conditioned gating component in \eqref{eq:placement_expert_weight}, i.e., $\mathbb E
\left[
\left(g_{\ell,i}^{(t)}\right)^2
\,\middle|\,
i\in S_\ell^{(t)}
\right]$.
    
    \item \emph{Random Placement} (RandPlace): 
    Experts are randomly assigned to devices under the one-expert-per-device constraint.

    \item \emph{Layer-Unaware Placement} (LayerUnaware): Experts are ranked according to the importance weight in \eqref{eq:placement_expert_weight} with $a_\ell=1$ for all MoE layers. This baseline is used to isolate the effect of calibrated layer-sensitivity weighting.
    
\end{itemize}
\begin{figure*}[t]
    \centering
    \subfloat[Inference accuracy\label{fig:layer_weight_acc}]{
        \includegraphics[width=0.30\textwidth]{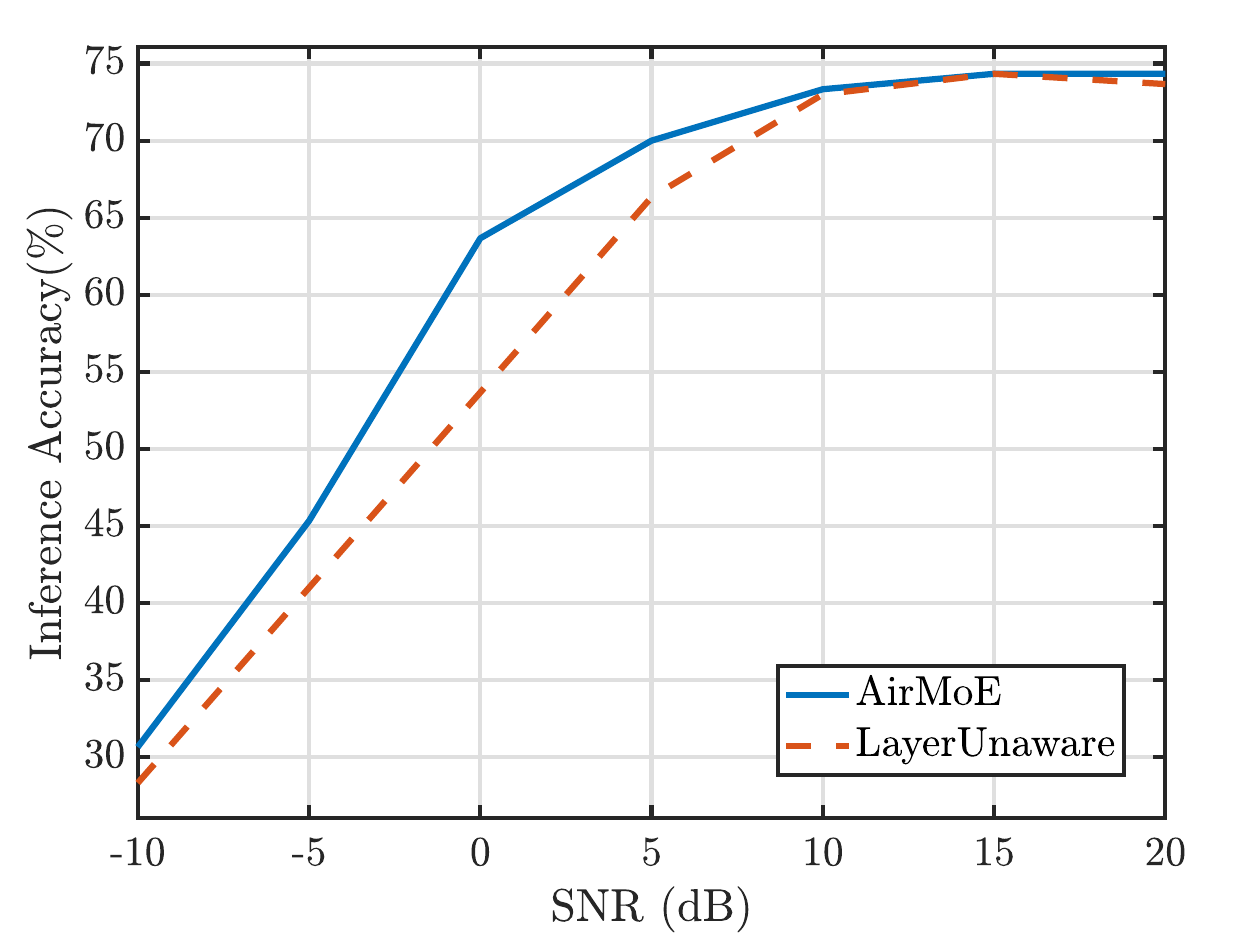}
    }
    \hfill
    \subfloat[Average layer-wise AirMoE error\label{fig:layer_weight_avg_mse}]{
        \includegraphics[width=0.30\textwidth]{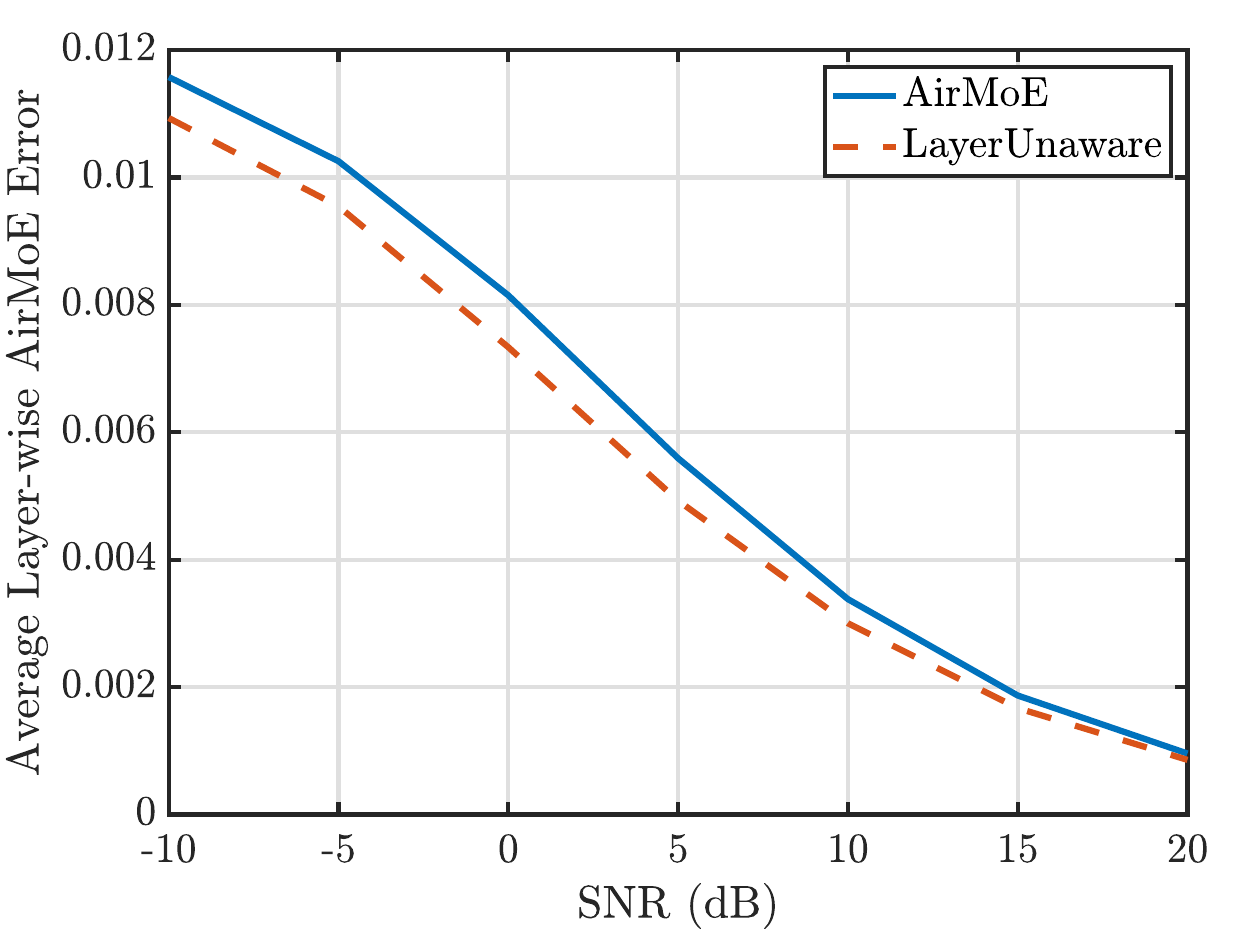}
    }
    \hfill
    \subfloat[AirMoE error\label{fig:layer_weight_weighted_mse}]{
        \includegraphics[width=0.30\textwidth]{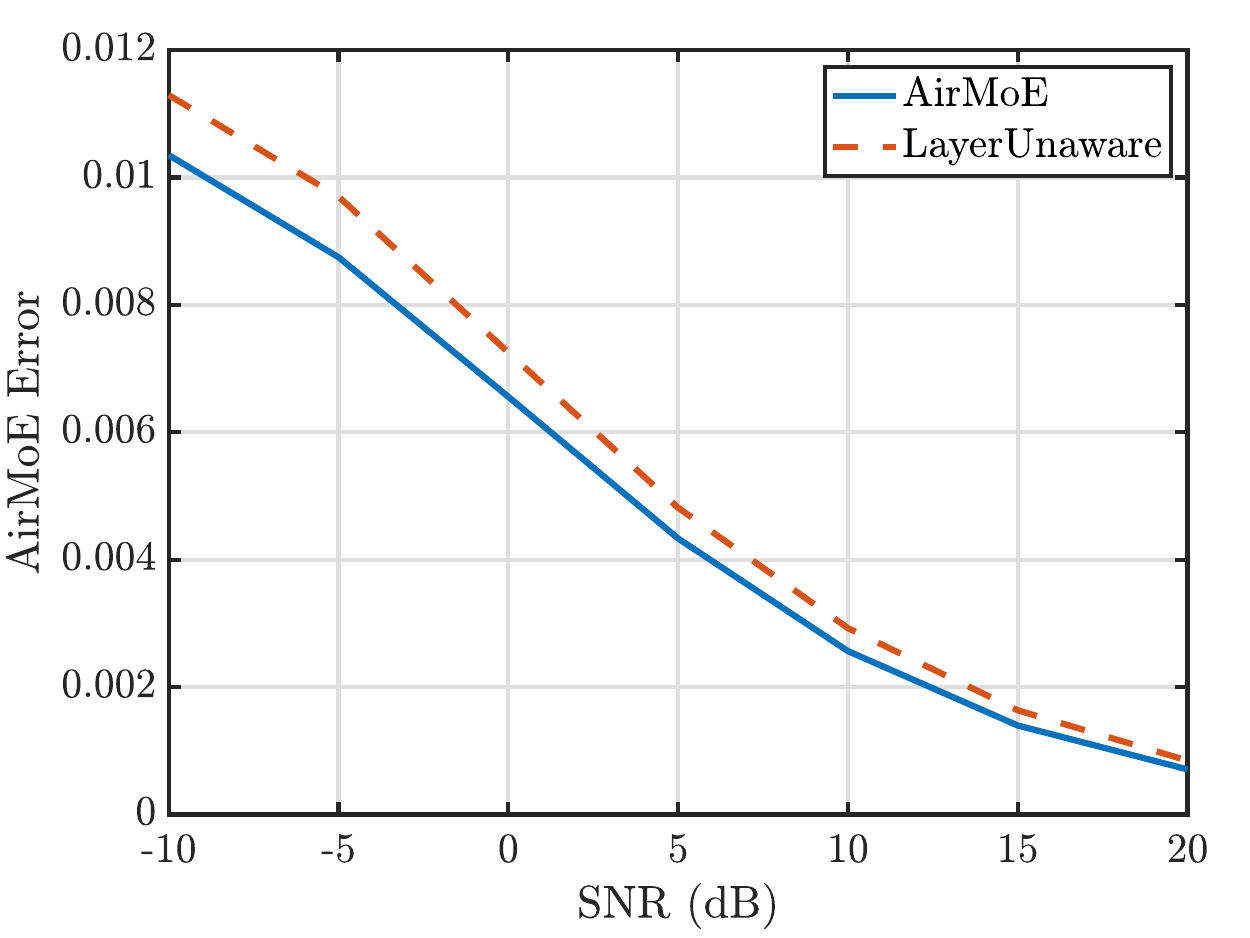}
    }
    \caption{Layer-sensitivity ablation under different SNRs.}
    \label{fig:ablation_studies}
    \vspace{-3mm}
\end{figure*}

\subsection{Gain from Power Control}
\label{subsec:gain_power_control}

Fig.~\ref{fig:arceasy_powercontrol_acc} shows the inference accuracy under different power-control policies with the proposed layer-aware placement. AirMoE achieves the highest accuracy over the considered SNR range, especially in the low- and medium-SNR regimes, demonstrating the benefit of jointly adapting the device-side precoders and the server-side denoising factor to instantaneous gating scores, channel realizations, and receiver noise in Theorem~\ref{thm:global_policy_fixedA}.

The trends of FullPower and ChannelInv are consistent with the asymptotic analysis in Sec.~\ref{subsec:asymptotic_analysis_extreme_snr}. At low SNRs, FullPower becomes relatively competitive because the receiver-noise term dominates the layer-wise AirMoE error. As the SNR increases, ChannelInv gradually approaches AirMoE, since coefficient mismatch becomes the dominant factor. In contrast, TruncInv performs poorly in the low-SNR regime because truncating weak-channel transmissions discards part of the router-selected expert outputs, resulting in an incomplete approximation of the ideal gating-weighted aggregation in \eqref{eq:moe_ideal_agg}.

\subsection{Gain from Expert Placement}
\label{subsec:gain_expert_placement}

Fig.~\ref{fig:arceasy_placement_acc} shows the inference accuracy under different expert-placement rules with the optimal power-control policy. AirMoE achieves the best overall performance, confirming the importance of the rank-preserving expert-device assignment characterized in Theorem~\ref{thm:placement_rank_preserving}.

The advantage of AirMoE comes from its more complete characterization of expert importance in \eqref{eq:placement_expert_weight}. ActFreqPlace only uses the activation frequency, while GatePlace only uses the activation-conditioned gating magnitude. Therefore, both baselines ignore part of the information captured by the proposed placement weight. RandPlace gives the lowest accuracy in most SNR regimes, further highlighting the need to match task-critical experts with devices of stronger long-term channel--power capability.

\begin{figure}[t]
    \centering
    \subfloat[Inference accuracy\label{fig:heterogeneity_acc}]{
        \includegraphics[width=0.45\linewidth]{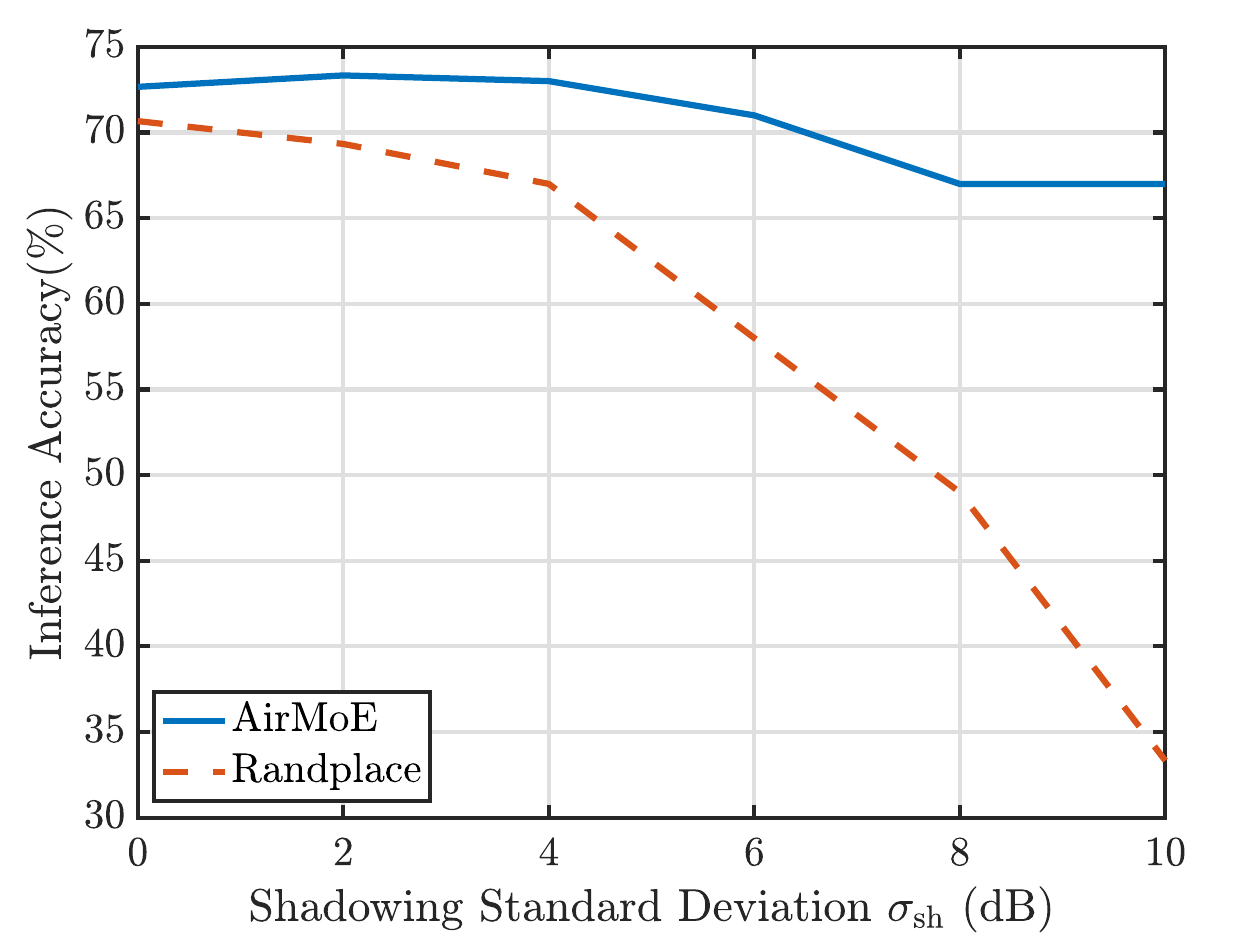}
    }
    \hfill
    \subfloat[AirMoE error\label{fig:heterogeneity_weighted_mse}]{
        \includegraphics[width=0.45\linewidth]{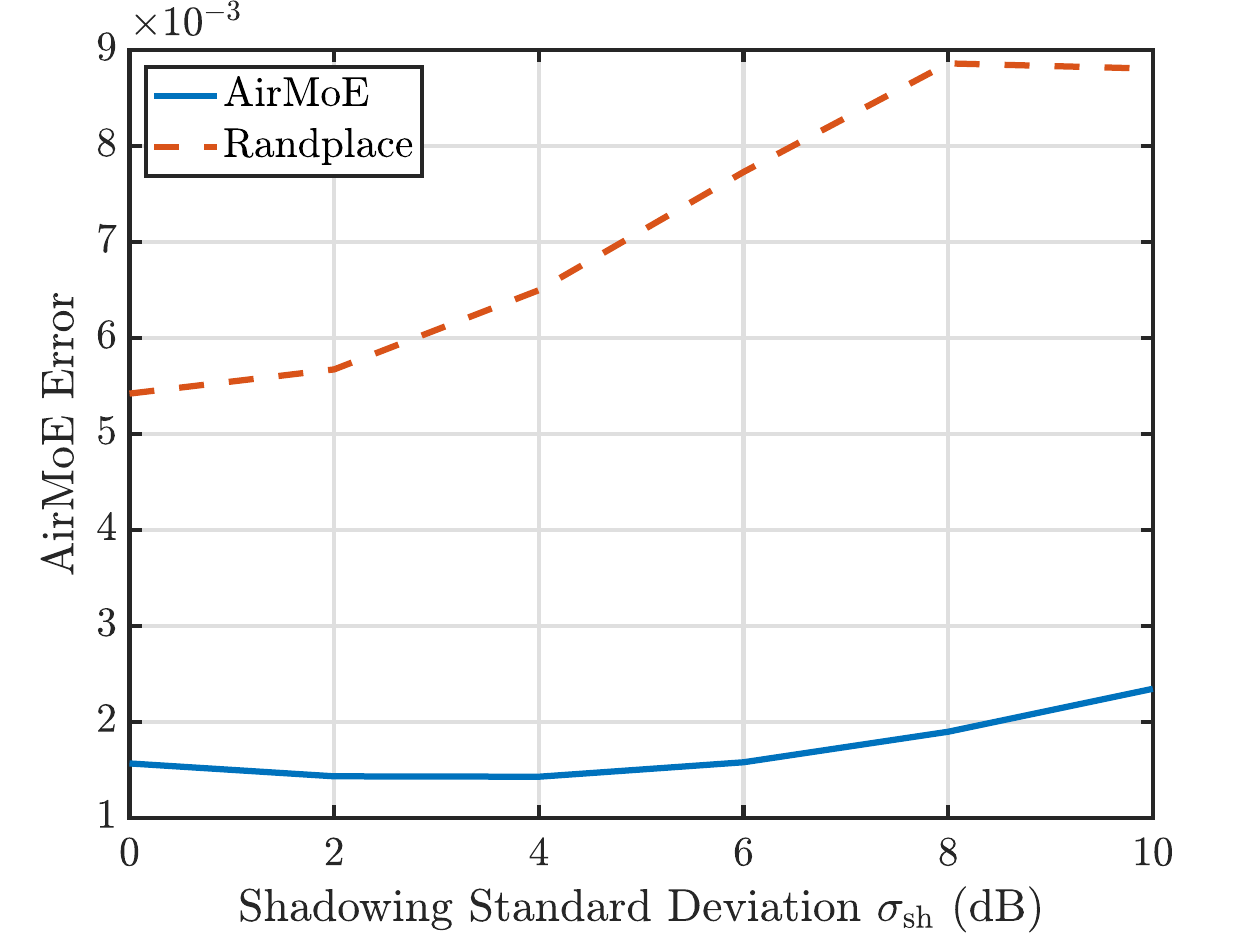}
    }
    \caption{Device-heterogeneity comparison under different shadowing standard deviations.}
    \label{fig:heterogeneity_ablation}
    \vspace{-3mm}
\end{figure}

\subsection{Effects of Layer Sensitivity}
\label{subsec:effect_layer_aware_weighting}

Fig.~\ref{fig:ablation_studies} evaluates the effect of the layer-sensitivity coefficients calibrated in Section~\ref{subsec:from_airmixer_error_to_inference_error} on expert placement and task-level distortion modeling. As shown in Fig.~\ref{fig:layer_weight_acc}, AirMoE achieves higher inference accuracy than LayerUnaware, especially in the low- and medium-SNR regimes. This is because aggregation distortion is more pronounced when the SNR is limited, making the layer where the distortion occurs more influential to the final prediction. Therefore, treating all MoE layers as equally sensitive to aggregation distortion can lead to suboptimal expert placement even under the same optimal power-control policy.

Figs.~\ref{fig:layer_weight_avg_mse} and~\ref{fig:layer_weight_weighted_mse} compare the unweighted average layer-wise AirMoE error with the proposed layer-sensitivity-weighted overall AirMoE error. The unweighted layer-wise average AirMoE error is not fully aligned with the accuracy trend: the layer-aware design achieves higher accuracy while having a slightly larger average aggregation error. This mismatch arises because the unweighted error averages distortions across layers equally and ignores their different inference-level impacts. By incorporating the calibrated layer-sensitivity coefficients, the proposed layer-weighted AirMoE error better matches the accuracy trend, confirming its role as a more task-relevant surrogate for E2E over-the-air MoE inference.

\subsection{Effects of Device Heterogeneity}
\label{subsec:effect_device_heterogeneity}

Fig.~\ref{fig:heterogeneity_ablation} examines how device heterogeneity affects expert placement.
We vary the shadowing standard deviation $\sigma_{\rm sh}$ to control the dispersion of long-term channel powers across devices. Fig.~\ref{fig:heterogeneity_acc} shows that the proposed placement and RandPlace achieve similar accuracy when $\sigma_{\rm sh}$ is small, since devices have comparable channel--power capabilities and the placement choice has limited impact. As $\sigma_{\rm sh}$ increases, the accuracy gap becomes much larger: RandPlace suffers a clear degradation, whereas the proposed placement remains robust. This is because stronger device heterogeneity amplifies the consequence of assigning important experts to weak devices, making rank-preserving expert-device matching more beneficial. Consistently, Fig.~\ref{fig:heterogeneity_weighted_mse} shows a widening AirMoE-error gap as $\sigma_{\rm sh}$ increases, explaining the stronger advantage of the proposed placement under higher device heterogeneity.

\section{Concluding Remarks}
\label{sec:conclusion}

This work has investigated over-the-air aggregation as a communication-efficient mechanism for wireless distributed MoE inference. The key observation is that the gating-weighted expert aggregation repeatedly performed in MoE layers is naturally aligned with wireless waveform superposition. Built on this alignment, AirMoE integrates simultaneous uplink expert-output aggregation with power control and expert placement to reduce inference-relevant aggregation distortion.
The resulting design highlights two principles for reliable over-the-air MoE inference: aggregation distortion should be evaluated according to its layer-dependent effect on inference performance, and expert placement should assign more inference-critical experts to devices with stronger long-term channel--power capability.
These principles suggest that wireless distributed inference should exploit the structure of neural model execution rather than treating communication as a separate transport layer.

From a broader perspective, AirMoE highlights several open problems in wireless MoE serving systems. Robust AirMoE aggregation can be revisited from an inference-aware perspective under imperfect CSI, synchronization errors, interference, and device mobility. Beyond empirical layer-sensitivity calibration, analytical models are needed to characterize how over-the-air aggregation distortion propagates across MoE layers and affects E2E inference performance. Another important direction is migration-aware adaptive expert placement, which balances the long-term inference benefit of updating expert-device associations against the communication, storage, and service-interruption costs incurred by expert migration.

\appendix

\subsection{Proof of Proposition~\ref{prop:problem_decomposition}}
\label{subsec:problem_decomposition}

Consider any placement $\mathbf X$ satisfying
\eqref{eq:placement_constraint_binary}--\eqref{eq:placement_constraint_device}.
For this fixed placement, the online power-control variables
$\{\mathbf p_{\ell}^{(t)},\eta_{\ell}^{(t)}\}$ are adapted to each instantaneous routing and channel realization.
Since the corresponding layer-wise AirMoE error and power-control constraints are decoupled across layer-token aggregations, the online variables can be optimized pointwise for each realization.
By the definition of $\mathrm{MSE}_{\ell}^{(t),\star}(\mathbf X)$ in \eqref{eq:min_mse}, any feasible online power-control choice satisfies
\begin{equation}
\label{eq:proof_pointwise_inequality}
\mathrm{MSE}_{\ell}^{(t)}
\left(
\mathbf X,\mathbf p_{\ell}^{(t)},\eta_{\ell}^{(t)}
\right)
\ge
\mathrm{MSE}_{\ell}^{(t),\star}(\mathbf X),
\end{equation}
for every layer-token aggregation and every realization.
Multiplying \eqref{eq:proof_pointwise_inequality} by $a_{\ell}\ge 0$, summing over $\ell$, and taking expectation yields
\begin{equation}
\label{eq:proof_lower_bound}
\sum_{\ell=1}^{L}
a_\ell
\mathbb E
\left[
\mathrm{MSE}_{\ell}^{(t)}
\left(
\mathbf X,\mathbf p_{\ell}^{(t)},\eta_{\ell}^{(t)}
\right)
\right]
\ge
\sum_{\ell=1}^{L}
a_\ell
\mathbb E
\left[
\mathrm{MSE}_{\ell}^{(t),\star}(\mathbf X)
\right].
\end{equation}
Thus, for the fixed placement $\mathbf X$, the right-hand side of \eqref{eq:proof_lower_bound} is a lower bound on the minimum achievable objective value over the online variables.

Conversely, because the online variables are allowed to adapt to each realization and are decoupled across instantaneous aggregations, selecting a pointwise optimal solution of \eqref{eq:min_mse} for each realization attains this lower bound.
Hence, for every feasible placement $\mathbf X$,
\begin{equation}
\begin{aligned}
\label{eq:proof_fixed_placement_value}
\min_{\{\mathbf p_{\ell}^{(t)},\eta_{\ell}^{(t)}\}}
\sum_{\ell=1}^{L}
a_{\ell}
\mathbb E
\left[
\mathrm{MSE}_{\ell}^{(t)}
\left(
\mathbf X,\mathbf p_{\ell}^{(t)},\eta_{\ell}^{(t)}
\right)
\right]=
\sum_{\ell=1}^{L}
a_{\ell}
\mathbb E
\left[
\mathrm{MSE}_{\ell}^{(t),\star}(\mathbf X)
\right],
\end{aligned}
\end{equation}
where the minimization on the left-hand side is subject to the instantaneous power-control constraints under the fixed placement $\mathbf X$.

Finally, minimizing \eqref{eq:proof_fixed_placement_value} over all placements satisfying
\eqref{eq:placement_constraint_binary}--\eqref{eq:placement_constraint_device}
yields the outer placement problem in \eqref{eq:min_placement}.
Combining an optimal placement of \eqref{eq:min_placement} with the corresponding pointwise optimal power-control variables from \eqref{eq:min_mse} therefore constructs an optimal solution to the original joint problem in \eqref{eq:joint_placement_power_problem}.
This proves the equivalence and completes the proof.

\subsection{Proof of Lemma~\ref{lem:eta_region_opt}}
\label{subsec:eta_region_opt}

For any given $n\in\{1,\dots,|\mathcal A|\}$, the first-order derivative of $\mathrm{MSE}_n(\eta)$ in \eqref{eq:mse_interval} with respect to $\eta$ is given by
\begin{equation}
\label{eq:mse_derivative}
\begin{aligned}
\frac{d}{d\eta}
\mathrm{MSE}_n(\eta)
=
\frac{2}{\eta^3}
\Bigg(
&\eta
\sum_{j=1}^{n}
|h_{(j)}|\sqrt{P_{(j)}}g_{(j)}
-
\sum_{j=1}^{n}
|h_{(j)}|^2P_{(j)}
-
\sigma^2
\Bigg).
\end{aligned}
\end{equation}
Since the term inside the parentheses in \eqref{eq:mse_derivative} is affine and strictly increasing in $\eta$, the stationary condition
$\frac{d}{d\eta}\mathrm{MSE}_n(\eta)=0$
admits a unique solution, given by $\hat{\eta}_n^\star$ in \eqref{eq:eta_n_stationary}.
Moreover, the derivative is nonpositive when $\eta\le \hat{\eta}_n^\star$ and nonnegative when $\eta\ge \hat{\eta}_n^\star$.
Therefore, $\mathrm{MSE}_n(\eta)$ is monotonically decreasing over $(0,\hat{\eta}_n^\star]$ and monotonically increasing over $[\hat{\eta}_n^\star,\infty)$.

It follows that the minimizer of $\mathrm{MSE}_n(\eta)$ over
$\mathcal R_n=[\tau_{(n)},\tau_{(n+1)}]$
is obtained by projecting $\hat{\eta}_n^\star$ onto $\mathcal R_n$, which gives \eqref{eq:eta_n_proj}.
This completes the proof.

\subsection{Proof of Lemma~\ref{lem:placement_layer_weighted_bound}}
\label{subsec:placement_layer_weighted_bound_proof}

Consider an arbitrary layer $\ell$, token $t$, and placement $\mathbf X$.
We construct the following feasible denoising factor:
\begin{equation}
\label{eq:eta_prime_placement_proof}
\eta_{\ell}^{\prime(t)}
=
\min_{i\in S_\ell^{(t)}}
\frac{
\sum_{m\in\mathcal M}
x_{\ell,i,m}
|h_m^{(t)}|\sqrt{P_m}
}
{
g_{\ell,i}^{(t)}
}.
\end{equation}
By construction, $\eta_{\ell}^{\prime(t)}$ is no larger than the saturation threshold of any activated expert in $S_\ell^{(t)}$, and hence all activated experts can achieve coefficient alignment.
Therefore,
\begin{equation}
\label{eq:placement_opt_upper_eta_prime}
\mathrm{MSE}_{\ell}^{(t),\star}(\mathbf X)
\le
\mathrm{MSE}_{\ell}^{(t)}(\eta_{\ell}^{\prime(t)};\mathbf X)
=
c_\ell^2
\frac{\sigma^2}{\left(\eta_{\ell}^{\prime(t)}\right)^2}.
\end{equation}

By the definition of $\eta_{\ell}^{\prime(t)}$, we have
\begin{equation}
\label{eq:eta_prime_inverse_bound}
\begin{aligned}
c_\ell^2
\frac{\sigma^2}{\left(\eta_{\ell}^{\prime(t)}\right)^2}
&=
c_\ell^2\sigma^2
\max_{i\in S_\ell^{(t)}}
\frac{
\left(g_{\ell,i}^{(t)}\right)^2
}
{
\left(
\sum_{m\in\mathcal M}
x_{\ell,i,m}
|h_m^{(t)}|\sqrt{P_m}
\right)^2
}
\\
&=
c_\ell^2\sigma^2
\max_{i\in S_\ell^{(t)}}
\sum_{m\in\mathcal M}
x_{\ell,i,m}
\frac{
\left(g_{\ell,i}^{(t)}\right)^2
}
{
|h_m^{(t)}|^2P_m
}
\\
&\le
c_\ell^2\sigma^2
\sum_{i\in S_\ell^{(t)}}
\sum_{m\in\mathcal M}
x_{\ell,i,m}
\frac{
\left(g_{\ell,i}^{(t)}\right)^2
}
{
|h_m^{(t)}|^2P_m
}.
\end{aligned}
\end{equation}
The second equality follows from the one-expert-per-device assignment constraint, under which only one $x_{\ell,i,m}$ equals one for each expert instance $(\ell,i)$. 
The last inequality follows from $\max_i a_i\le \sum_i a_i$ for non-negative $\{a_i\}$.

Combining \eqref{eq:placement_opt_upper_eta_prime} and \eqref{eq:eta_prime_inverse_bound}, we obtain
\begin{equation}
\label{eq:placement_layer_token_bound}
\mathrm{MSE}_{\ell}^{(t),\star}(\mathbf X)
\le
c_\ell^2\sigma^2
\sum_{i\in S_\ell^{(t)}}
\sum_{m\in\mathcal M}
x_{\ell,i,m}
\frac{
\left(g_{\ell,i}^{(t)}\right)^2
}
{
|h_m^{(t)}|^2P_m
}.
\end{equation}
Taking expectation over tokens and channel realizations, multiplying by $a_\ell$, and summing over all layers yield \eqref{eq:placement_factorized_bound_raw}.
This completes the proof.

% \subsection{Proof of Theorem~\ref{thm:placement_rank_preserving}}
% \label{subsec:rank_preserving_placement_proof}

% We prove the result by an exchange argument.
% Consider any feasible placement that contains an inversion, i.e., two expert instances with weights
% $W_a\ge W_b$ are assigned to two devices with costs $R_p\ge R_q$, respectively.
% Swapping the two assigned devices preserves feasibility, since it only exchanges the devices assigned to two expert instances.

% The change in the objective value after this swap is
% \begin{equation}
% \begin{aligned}
% \Delta
% &=
% \left(W_aR_q+W_bR_p\right)
% -
% \left(W_aR_p+W_bR_q\right) \\
% &=
% -\left(W_a-W_b\right)\left(R_p-R_q\right)
% \le 0 .
% \end{aligned}
% \end{equation}
% Therefore, removing an inversion cannot increase the objective value.
% Starting from any feasible placement, repeatedly applying this exchange step eliminates all inversions between the ordering of expert weights and the ordering of device costs.
% The resulting placement is exactly the rank-preserving assignment in \eqref{eq:rank_preserving_assignment}, and is therefore globally optimal for \eqref{eq:placement_separable_problem}.

\subsection{Proof of Theorem~\ref{thm:placement_rank_preserving}}
\label{subsec:rank_preserving_placement_proof}

We prove the result by an exchange argument.
First, if a used device $p$ has a larger cost than an unused device $q$, i.e., $R_p>R_q$, reassigning the expert on device $p$ to device $q$ preserves feasibility and does not increase the objective since $W_j\ge0$.
Hence, an optimal solution uses the $|\mathcal J|$ devices with the smallest costs.

It remains to determine the matching among these devices.
Consider two experts with $W_a\ge W_b$ assigned to devices with $R_p\ge R_q$, respectively.
Swapping their assignments changes the objective by
\begin{equation}
\begin{aligned}
\Delta
&=
\left(W_aR_q+W_bR_p\right)
-
\left(W_aR_p+W_bR_q\right)\\
&=
-\left(W_a-W_b\right)\left(R_p-R_q\right)
\le0.
\end{aligned}
\end{equation}
Thus, removing any inversion cannot increase the objective.
Repeatedly applying this exchange yields an optimal rank-preserving assignment, where the $j$-th largest-weight expert is assigned to the $j$-th lowest-cost device, as in \eqref{eq:rank_preserving_assignment}.

\IfFileExists{main.bib}
  {\bibliography{main}}
  {\bibliography{../../main}}
\bibliographystyle{IEEEtran}

\end{document}